\documentclass[american,aps,pra,reprint,superscriptaddress,longbibliography,onecolumn]{revtex4-2}

\usepackage{amsmath,amssymb}

\usepackage{color}
\usepackage{graphicx}
\usepackage[american]{babel}
\usepackage[utf8]{inputenc}
\usepackage{times}
\usepackage{braket} 				
\usepackage{amsthm}
\usepackage{amsfonts}
\usepackage{mathrsfs}  
\usepackage{booktabs}
\usepackage{pgfplots}
\pgfplotsset{compat=1.17}
\usepackage{mathtools}
\usepackage{bbold}
\usepackage{titlesec}
\usepackage{bbm}

\newcommand{\ketbra}[2]{|#1\rangle\!\langle#2|}

\definecolor{mygrey}{gray}{0.35}
\definecolor{myblue}{rgb}{0.2,0.2,0.8}
\definecolor{myzard}{cmyk}{0,0,0.05,0}
\definecolor{mywhite}{rgb}{1,1,1}
\definecolor{myred}{rgb}{0.9,0.1,0.}
\usepackage[colorlinks=true,citecolor=myblue,linkcolor=myblue,urlcolor=myblue]{hyperref}

\newtheoremstyle{customStyle1}  
{0pt}       
{0pt}       
{\normalfont}   
{\parindent}        
{\em}  
{. --}   	 
{.5em}       
{\thmname{#1}\thmnumber{ #2}\thmnote{ (#3)}}  

\newcounter{theorems}
\newtheorem{thm}[theorems]{Theorem}
\newtheorem{prop}[theorems]{Proposition}

\newtheorem{definition}[theorems]{Definition} 
\newtheorem{lem}[theorems]{Lemma}

\renewcommand\vec{\mathbf}

\newcommand{\p}{\vec{p}}
\newcommand{\q}{\vec{q}}

\DeclareMathOperator{\tr}{Tr}
\newcommand{\norm}[1]{\left\lVert#1\right\rVert}

\DeclareMathOperator{\TO}{TO} 
\DeclareMathOperator{\CTO}{CTO} 
\DeclareMathOperator{\GPO}{GPO} 
\DeclareMathOperator{\Prob}{Prob} 

\begin{document}
\title{Thermodynamic state convertibility is determined by qubit cooling and heating}
	\author{Thomas Theurer}
	\email{thomas.theurer@ucalgary.ca}
	\author{Elia Zanoni}
	\author{Carlo Maria Scandolo}
	\author{Gilad Gour}
	\affiliation{Department of Mathematics and Statistics, University of Calgary, Calgary, AB T2N 1N4, Canada}
	\affiliation{Institute for Quantum Science and Technology, University of Calgary, Calgary, AB T2N 1N4, Canada}
\begin{abstract}
	Thermodynamics plays an important role both in the foundations of physics and in technological applications. An operational perspective adopted in recent years is to formulate it as a quantum resource theory. 
	At the core of this theory is the interconversion between \textit{athermality} states, i.e., states out of thermal equilibrium. Here, we solve the question of how athermality can be used to heat and cool other quantum systems that are initially at thermal equilibrium. We then show that the convertibility between quasi-classical resources (resources that do not exhibit coherence between different energy eigenstates) is fully characterized by their ability to cool and heat qubits, i.e., by two of the most fundamental thermodynamical tasks on the simplest quantum systems.
\end{abstract}
\date{\today}
\maketitle

\section{Introduction}
The constant increase of control that experiments exert on small-scale quantum devices has led to a growing interest in heat engines operating in the quantum regime, potentially granting quantum advantages over their classical counterparts~\cite{Scovil1959, Faucheux1995, Scully2002, Baugh2005, Uzdin2015, Silva2016, Klatzow2019, Bera2021}. In contrast, quantum effects hinder the miniaturization of classical devices by introducing new sources of errors.
To apprehend this dual role that quantum mechanics plays in nanotechnology, understanding thermodynamics in the quantum limit is of high importance. 
It is also important from a foundational perspective since quantum thermodynamics has applications ranging from the small scales of biochemistry to the large scales of black hole physics~\cite{Gemmer2009, Binder2018, YungerHalpern2020, Spaventa2022}. 

An operational perspective on quantum thermodynamics that was adopted in recent years~\cite{Janzing2000, Aberg2013, Brandao2013, Skrzypczyk2014, Egloff2015, Horodecki2013, Brandao2015, Lostaglio2015, Lostaglio2015b, Goold2016, Korzekwa2016, Masanes2017, Scharlau2018, Lostaglio2019, Faist2019, Faist2021} is to study it in the framework of quantum resource theories~\cite{Coecke2016, Chitambar2019}. Such theories emerge from physically-motivated restrictions that are imposed on top of the laws of quantum mechanics.
This singles out the resource under consideration and, in the case of resource theories concerned with the value of quantum states such as the ones that we consider in the following, divides both states and channels into free and resourceful. In the case of thermodynamics, every (non-interacting) system has exactly one free state, namely its Gibbs state corresponding to a fixed background temperature. Empirically, this is the state to which the system evolves naturally and from which no work can be extracted~\cite{Uffink2001, Brown2001}, motivating the notion ``free state''. The free operations must map free states to free states, allowing for resources to be manipulated but not freely created. 
Here, we consider (the closure of) the thermal operations~\cite{Janzing2000, Brandao2013} as free, which can be implemented with unitaries that are energy-preserving and potentially act on parts of a bath at the fixed background temperature too. Any source of athermality, such as, e.g., a single qubit out of thermal equilibrium, is then considered a resource and described explicitly. This allows us to investigate thermodynamics in the quantum limit. 
Once the set of free states and operations is fixed, a resource theory studies how resources, in our case states out of thermal equilibrium, can be utilized to perform tasks of practical relevance. Whenever one resource can be converted to another using only free operations, this means that in \textit{any} application, the latter cannot be more valuable than the former: Assume for example that we can use thermal operations and the latter to perform work. Then we can perform an equal amount of work using thermal operations and the former by first converting the former to the latter (which is free) and then using the same (free) protocol. However, in principle, we might be able to perform more work using a different free protocol. The convertibility of resources is thus at the core of all resource theories since it describes their relative value and thus usefulness for applications. This also justifies the recent efforts to geometrically characterize the states that can both be reached from and converted to a given state using thermal operations~\cite{Lostaglio2018, Mazurek2018, Mazurek2019, Oliveira2022, Ende2022b, Ende2022c}.

Two of the most notable thermodynamical tasks are cooling and heating, important for (quantum) computation~\cite{Landauer1961} when initializing (qu)bits and countless other examples. In this work, we always assume that the system $A$ that we intend to cool or heat is initially in its free, i.e., thermal equilibrium, state. Since this is the state to which systems naturally evolve, this is a relevant scenario. Moreover, as expected, thermalizing systems is free in our framework, allowing us to bring all systems to their equilibrium state before we start cooling or heating them. If the system we intend to cool/heat is thus initially not in its free state, any results derived under this assumption can be seen as an upper bound on the cost of cooling/heating. Cooling $A$ can then have several meanings~\cite{Clivaz2019}. One interpretation is to bring $A$ to a Gibbs state corresponding to a lower temperature. 
Another common interpretation that is independent of a notion of temperature is to increase the overlap with $A$'s ground state, which is considered for example in algorithmic heat bath cooling~\cite{Schulman1999, Boykin2002, Fernandez2004, Baugh2005, Taranto2020}. In this work, we solve the question of how much we can cool a system initially at equilibrium when having access to a given athermality state and the closure of the thermal operation for \textit{either} interpretation, see Fig.~\ref{fig:TO}. Importantly, when cooling qubits, the two interpretations are equivalent and reduce to increasing the ground state population. Similarly, heating a qubit reduces to increasing the population of its higher energy level. 
\begin{figure}[tb!]
	\centering
	\includegraphics[width=0.4\linewidth]{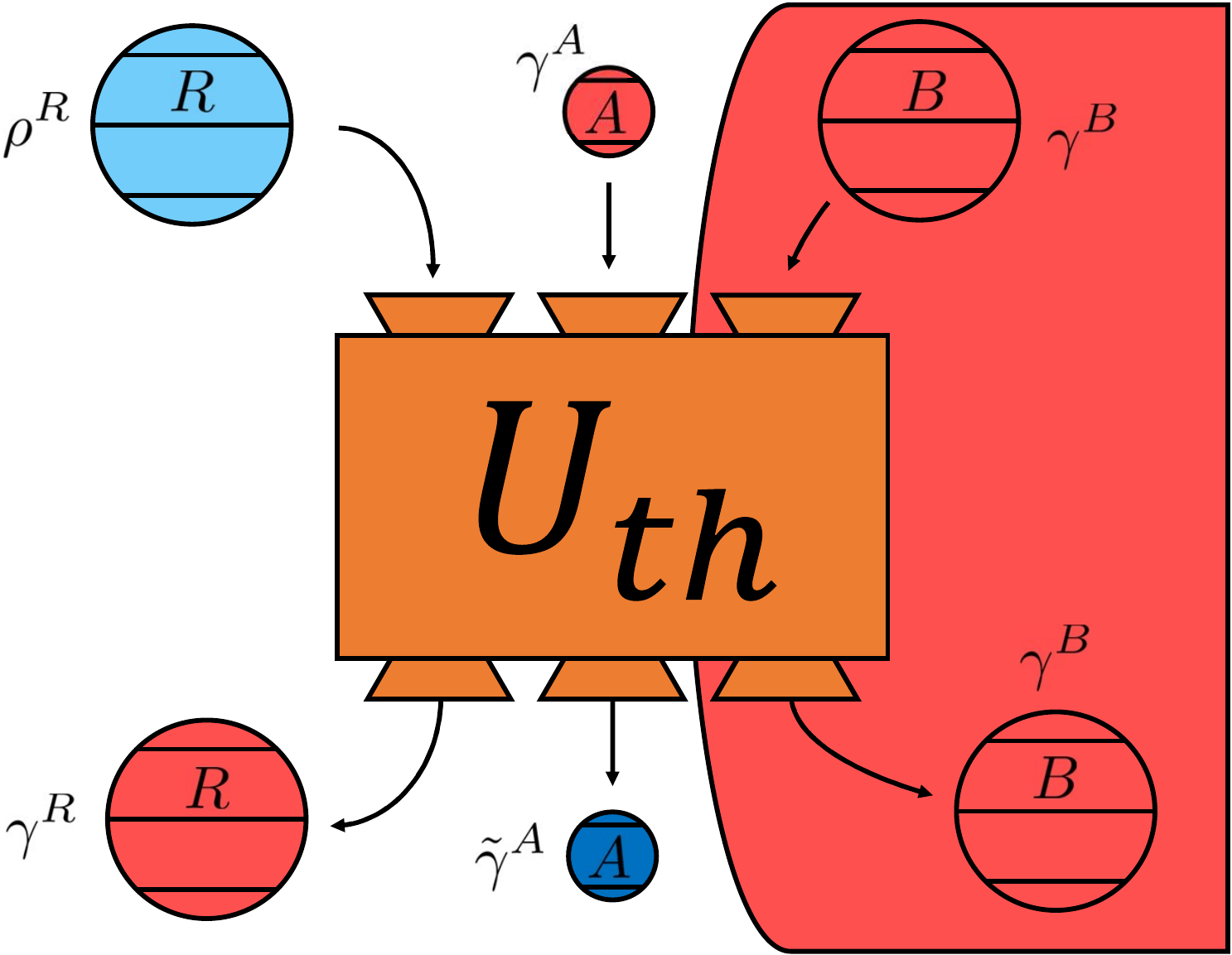}
	\caption{Consider a bath with a fixed background temperature (red). Given a system $R$ that is initially in a non-equilibrium state $\rho^R$, how much can we cool a system $A$ initially in the equilibrium state $\gamma^A$  using thermal operations implemented by a unitary $U_{th}$ that can also act on parts of the bath $B$ (at equilibrium), but needs to preserve global energy?} \label{fig:TO}
\end{figure}

As our main result, we show that the ability to cool and heat qubits fully characterizes the convertibility between quasi-classical states, i.e., states that do not exhibit coherence between different energy eigenstates, and therefore connect two of the central questions in the resource theory of quantum thermodynamics. Quasi-classical states are an important subsets of states out of thermal equilibrium, since typically superpositions between different energy eigenstates decay much quicker than the populations thermalize. For proofs and further information, see the Appendices.

\section{Notation and preliminaries}\label{sec:notation}
In this work, we restrict our considerations to finite-dimensional quantum systems corresponding to finite-dimensional Hilbert spaces. We denote quantum systems by capital Latin letters such as $A$ and their dimension by $|A|$. 
Quantum states are represented by small Greek letters, e.g., $\rho^A$, where the superscript indicates that $\rho$ is a state of system $A$. The symbol $\gamma$ is reserved for Gibbs states. 
Quantum channels, i.e., completely positive and trace-preserving maps (CPTP), are denoted by calligraphic capital Latin letters such as  $\mathcal{E}^{{B}\leftarrow A}$, where the superscript indicates that $\mathcal{E}$ maps states of system $A$ to states of system $B$. 

Real valued vectors are represented by small bold Latin letters such as $\vec{r}$, which has $r_x$ as its $x$-th entry, and we write $\Prob(n)$ for the set of probability vectors of dimension $n$. 
We frequently utilize the Ky-Fan norms defined on $\mathbb{R}^n$ as 
\begin{align}\label{eq:KyFanNorm}
	\norm{\p}_{(k)}:=\sum_{x=1}^k p_x^\downarrow\quad \forall k\in[n],
\end{align}
where $[n]$ is a shorthand notation for $\{1,2,3,\cdots,n\}$, and  $\vec{p}^\downarrow$ is the vector with the absolute values of the entries of $\p$ arranged in non-increasing order.

A technical tool with many applications and generalizations in the quantum regime~\cite{Nielsen1999, Winter2016, Zhu2017, Buscemi2017, Gour2018, Chubb2018, Horodecki2018, Mueller2018, Korzekwa2019, Buscemi2019, Ende2020, DallArno2020, Singh2021, Koukoulekidis2022} on which we rely is the concept of (relative) majorization: 
Let $\p,\vec{r}\in\Prob(n)$ and $\q,\vec{s}\in\Prob(m)$. We say that $(\p,\vec{r})$ relatively majorizes $(\q,\vec{s})$, written as $(\p,\vec{r})\succ(\q,\vec{s})$, iff there exists a $m\times n$ column stochastic matrix $E$ such that 
\begin{align}\label{eq:relMajo}
	E\vec{p}=\vec{q},\quad E\vec{r}=\vec{s}.
\end{align}
Relative majorization can also be characterized in an appealing geometrical manner: If $\pi(j)$ is a permutation such that the sequence $\{p_{\pi(j)}/r_{\pi(j)}\}_{j=1}^n$ is non-increasing, one defines for $k\in[n]$
\begin{align}\label{eq:elbows}
	(x_k^\star(\vec{p},\vec{r}),y_k^\star(\vec{p},\vec{r}))=\left(\sum_{j=1}^kp_{\pi(j)},\sum_{j=1}^kr_{\pi(j)}\right).
\end{align}
These ``elbows" and $(x_0^\star,y_0^\star)=(0,0)$ connected by straight lines form the lower boundary of the so-called testing region of $(\p,\vec{r})$. 
Now let this lower boundary be described by the function $y \mapsto\alpha_y(\vec{p},\vec{r})$, i.e., we describe the $x$-values of the lower boundary as a function of the $y$-values.
Importantly, it can then be shown that $(\p,\vec{r})\succ(\q,\vec{s})$ iff $\alpha_y(\vec{p},\vec{r})\ge \alpha_y(\vec{q},\vec{s})\ \forall y\in[0,1]$~\cite{Blackwell1953,Uhlmann1978,Ruch1978}, which in turn is equivalent to $m$ inequalities being satisfied simultaneously~\cite{Renes2016}.

\section{Thermal and closed thermal operations}\label{sec:thermalOp}
In this section, we introduce the basic building blocks of resource theories of thermodynamics as needed for this work. 
The Gibbs state of a system $A$ with Hamiltonian $H^A$ and at inverse temperature $\beta$ is defined as 
\begin{align}
	\gamma^A(\beta)=\frac{e^{-\beta H^A}}{\tr\left(e^{-\beta H^A}\right)}.
\end{align}	
In the following, we will suppress the dependency on $\beta$ if clear from the context. As mentioned in the introduction, only Gibbs states corresponding to the fixed background temperature are considered free. Discussing the resourcefulness of a given state $\rho^A$ is therefore only meaningful with respect to this reference. This is comparable to, e.g., the various resource theories of coherence~\cite{Aberg2006, Baumgratz2014, Streltsov2017}, which require us to fix an incoherent basis, or entanglement theories~\cite{Plenio2007, Horodecki2009}, that rely on a notion of spatial separation. Throughout this work, whenever we have a state $\rho^A$, we thus write the \textit{athermality state} $(\rho^A,\gamma^A)$ to make it clear that we consider it with respect to the free Gibbs state $\gamma^A$. 
From a physics perspective, it is meaningful to consider only non-zero temperatures and finite energies, which we will do in the following. This implies that all Gibbs states are full rank.

To build a resource theory, in addition to the free states, we need free operations. One possible set of free operations are the thermal operations (TO), see Refs.~\cite{Janzing2000, Brandao2013}.
\begin{definition}\label{def:TO}
	Let $A, A'$ be fixed physical systems with free Gibbs states $\gamma^A, \gamma^{A'}$ and $B$, $B'$ arbitrary (finite-dimensional) physical systems with free Gibbs states $\gamma^{B}, \gamma^{B'}$  such that $AB=A'B'$ and $\gamma^{AB}=\gamma^A\otimes \gamma^B= \gamma^{A'}\otimes \gamma^{B'}$. The set of CPTP maps of the form 
	\begin{align}
		\mathcal{E}^{{A'}\leftarrow A}(\rho^A)=\tr_{B'}\left[U^{AB}\left(\rho^A\otimes\gamma^B\right)U^{\dagger AB}\right],
	\end{align}
	where $U^{AB}$ is a unitary that commutes with $\gamma^{AB}$, forms the set of thermal operations (from system $A$ with Gibbs state $\gamma^A$ to system $A'$ with Gibbs state $\gamma^{A'}$). 
\end{definition}
Since the definition of the thermal operations is also relative to two Gibbs states, we want to make this apparent too. Note however that the Gibbs state of the target system is determined by the Gibbs state of the initial system and the operation since 
\begin{align}\label{eq:TOpresGibbs}
	\mathcal{E}^{{A'}\leftarrow A}(\gamma^A)=&\tr_{B'}\left[U^{AB}\left(\gamma^A\otimes\gamma^B\right)U^{\dagger AB}\right] =\gamma^{A'}.
\end{align}
We therefore write $(\mathcal{E}^{{A'}\leftarrow A},\gamma^A)$ and denote the set of thermal operations from system $A$ with Gibbs state $\gamma^A$ to system $A'$ with Gibbs state $\gamma^{A'}$ as TO$(\gamma^{A'}\leftarrow\gamma^{A})$. 	
Note that in the above definition, we required that the Gibbs states of the considered composed systems are of product form, i.e., $\gamma^{AB}=\gamma^A\otimes \gamma^B= \gamma^{A'}\otimes \gamma^{B'}$. This implies that the subsystems are non-interacting, which ensures that they possess a well-defined thermal equilibrium state. If we had interaction terms, a thermal equilibrium state would only be defined for the combined system, rendering the concept of a free state of a subsystem meaningless~\cite{Chiribella2017}. A thermal operations thus consists of three steps: We bring our initial system in contact with another non-interacting system at equilibrium, perform a global energy preserving unitary, and discard a non-interacting subsystem.

Note that in Def.~\ref{def:TO}, the dimension of system $B$ is in general unbounded (although finite due to our restriction that we only consider finite-dimensional systems). In principle, it is thus possible that a given state $(\rho^A,\gamma^A)$ cannot be converted to $(\sigma^{A'},\gamma^{A'})$ via thermal operations, whilst one can approximate $(\sigma^{A'},\gamma^{A'})$ arbitrarily well using thermal operations and $(\rho^A,\gamma^A)$ (see Ref.~\cite[p.4 bottom]{Ende2022} for an explicit example). However, from a practical perspective, it is impossible to distinguish $(\sigma^{A'},\gamma^{A'})$ from an arbitrarily close approximation, and the same holds for channels. 
It is therefore meaningful to consider the topological closure of TO$(\gamma^{A'}\leftarrow\gamma^{A})$, which we denote by CTO$(\gamma^{A'}\leftarrow\gamma^{A})$, see also Ref.~\cite[Def.~7]{Janzing2000}, Refs.~\cite{Ende2022,Gour2022}, and the App.~\ref{app:NotRem} and App.~\ref{app:PropCTO}. We will call this the set of closed thermal operations ($\CTO$). Whilst it is not always stated explicitly, in the literature, CTO is in fact the common choice for the free operations. This includes, e.g., Refs.~\cite{Horodecki2013,Brandao2015}. To avoid potential confusion, we decided to state explicitly that we use CTO.

So far, we have discussed two slightly different sets of free operations, namely $\TO$ and $\CTO$, both with respect to the same free states. 
A channel that is contained in both sets of free operations~\cite[Lem.~II.2]{Gour2022}  is the twirling or pinching channel: expressing $\gamma$ as
\begin{align}
	\gamma=\sum_{x=1}^{m}a_x \Pi_x
\end{align}
where $\{a_x\}_{x=1}^m$ are distinct eigenvalues, and $\{\Pi_x\}_{x=1}^m$ its spectral projectors, the pinching channel (with respect to $\gamma$) is given by 
\begin{align}
	\mathcal{P}_{\gamma}(\rho)=\sum_{x=1}^m\Pi_x\rho\Pi_x.
\end{align}
It plays an important role when we talk about \textit{quasi-classical} states, i.e., athermality states $(\rho,\gamma)$ with $[\rho,\gamma]=0$ since such states are invariant under the application of the pinching channel. Due to the commutation condition, quasi-classical states are fully characterized by a corresponding pair of vectors $\vec{r},\vec{g}$ that contain the (consistently ordered) eigenvalues of $\rho,\gamma$. To conclude this section, we note that $(\rho^R,\gamma^R) \xrightarrow{\CTO} (\sigma^S,\gamma^S)$ is used to express that $(\rho^R,\gamma^R)$ can be converted to $(\sigma^S,\gamma^S)$ using closed thermal operations, i.e., that there exists an $\mathcal{E}\in\CTO\left(\gamma^S\leftarrow\gamma^R\right)$ with $\mathcal{E}(\rho^R)=\sigma^S$.

\section{Cooling and heating}
In the following, we investigate the cooling and heating of a system $A$ initially in its free state and with Hamiltonian $H^A$ which is not completely degenerate, since that would render the problem trivial. We begin with the first interpretation of cooling, i.e., we consider transformations of the form $(\rho^R\otimes\gamma^A,\gamma^R\otimes\gamma^A)\xrightarrow{\CTO}(\tilde{\gamma}^A,\gamma^A)$ where $(\rho^R,\gamma^R)$ is a fixed initial athermality state and $\tilde{\gamma}^A$ a Gibbs state corresponding to a temperature $\tilde{\beta}$ potentially different from the background temperature $\beta$ that defines $\gamma^A$. Since preparing a system in its equilibrium state is free, we can write this more compactly as
\begin{align}
	(\rho^R,\gamma^R)\xrightarrow{\CTO}(\tilde{\gamma}^A,\gamma^A).
\end{align} 

\begin{thm}\label{thm:CoolingHeating}
	Consider an initial resource $(\rho^R,\gamma^R)$ and let $\vec{r}^R, \vec{g}^R, \tilde{\vec{g}}^A(\tilde{\beta}),\vec{g}^A$ denote the probability vectors corresponding to $\mathcal{P}_{\gamma^R}(\rho^R), \gamma^R,\tilde{\gamma}^A(\tilde{\beta}),\gamma^A$, respectively, where the components of $\tilde{\vec{g}}^A(\tilde{\beta})$ are assumed to be ordered such that $\tilde{g}_{j}(\tilde{\beta}_k)$ corresponds to the $j$-th smallest eigenenergy of system $A$. Let further be $\alpha$ as introduced below Eq.~\eqref{eq:elbows}. 	The maximal $\tilde{\beta}$ to which we can cool system $A$ using $\CTO$s is then given by 
	\begin{align}\label{eq:minT}
		\tilde{\beta}_{\max}=\min_{k\in[|A|-1]}\left\{\tilde{\beta}_k: \norm{\tilde{\vec{g}}^A(\tilde{\beta}_k)}_{(k)}=\alpha_{ \norm{\vec{g}^A}_{(k)}}(\vec{r}^R,\vec{g}^R)\right\},
	\end{align}
	and the minimal (potentially negative) $\tilde{\beta}$ to which we can heat system $A$ by
	\begin{align}
		\tilde{\beta}_{\min}=\max_{k\in[|A|-1]}\!&\Biggl\{\!  \tilde{\beta}_k:\!  \sum_{j=|A|-k+1}^{|A|}\! \tilde{g}^A_{j}(\tilde{\beta}_k)\! =\! \alpha_{1- \norm{\vec{g}^A}_{(|A|-k)}}(\vec{r}^R,\vec{g}^R) \Biggr\}.
	\end{align}
	For qubits with energy gap $E$, this simplifies to
	\begin{align}\label{eq:TminQubit}
		\tilde{\beta}_{\max}=&\frac{1}{E}\ln\left(\frac{\alpha_{ g_1^A}(\vec{r}^R,\vec{g}^R)}{1-\alpha_{ g_1^A}(\vec{r}^R,\vec{g}^R)}\right),  \\
		\tilde{\beta}_{\min} =&\frac{1}{E} \ln\left(\frac{1-\alpha_{g_2^A}(\vec{r}^R,\vec{g}^R)}{\alpha_{g_2^A}(\vec{r}^R,\vec{g}^R)}\right).
	\end{align}
\end{thm}
The above Theorem provides closed-form expressions for the lowest and highest $\tilde{\beta}$ to which we can heat or cool a qubit, respectively, since $\alpha_y(\vec{r}^R,\vec{g}^R)$ is obtained from interpolating between the elbows corresponding to $(\vec{r}^R,\vec{g}^R)$ (see the proof in App.~\ref{app:Proofs} for more information). For general target systems, the equations defining the $\tilde{\beta}_k$ both for cooling and heating need to be solved numerically. This is however straightforward, since, e.g., $||\tilde{\vec{g}}^A(\tilde{\beta}_k)||_{(k)}$ is a monotonic function of $\tilde{\beta}_k$ (it is just the probability to find the system in one of its $k$ lowest energy eigenstates), and the right-hand sides of the equations are independent of $\tilde{\beta}_k$. We also allowed for negative $\tilde{\beta}$, which corresponds to population inversions~\cite{Svelto1998,Kardar2007}. This will be of use later. As apparent form the proof of the Theorem, every $\tilde{\beta}$ between $\tilde{\beta}_{\min}$ and $\tilde{\beta}_{\max}$ is reachable too. Interestingly, in Thm.~\ref{thm:CoolingHeating}, only $\mathcal{P}_{\gamma^R}(\rho^R)$ appears, but not $\rho^R$ itself. This implies that for cooling and heating, $(\mathcal{P}_{\gamma^R}(\rho^R),\gamma^R)$ is exactly as useful as $(\rho^R,\gamma^R$). In other words, coherences between different energy eigenstates are irrelevant. The reason for this is that thermal operations are covariant~\cite{Lostaglio2015}, see Lem.~10 in the Appendix for more information. A consequence is that a state $(\rho^R,\gamma^R)\ne (\gamma^R,\gamma^R)$ but $(\mathcal{P}_{\gamma^R}(\rho^R),\gamma^R)= (\gamma^R,\gamma^R)$ is non-free, but useless for cooling or heating. Note that Ref.~\cite{Janzing2000} also treats the cooling and heating of qubits in Thms.~3 and 7, but does not derive a closed form expression for the single-shot case.

Briefly touching on the second interpretation of cooling, we maximize the overlap with the ground state. The maximal overlap $O_{\max}$ with a potentially degenerate ground state that we can achieve given access to $\CTO$, an initial resource $(\rho^R,\gamma^R)$, and Gibbs state $\gamma^A$, i.e., 
\begin{align}\label{eq:OmaxDef}
	O_{\max}(\rho^R,\gamma^R,\gamma^A)
	=&\max\{\tr\left[\Pi^A\tau^A\right]: (\rho^R,\gamma^R)\xrightarrow{\CTO} (\tau^A,\gamma^A) \},
\end{align}
where $\Pi^A$ is the projector onto the (degenerate) ground state, is given in the following Proposition.
\begin{prop}\label{prop:altCooling}
	With $\vec{r}^R, \vec{g}^R,\vec{g}^A$ denoting the probability vectors corresponding to $\mathcal{P}_{\gamma^R}(\rho^R), \gamma^R,\gamma^A$, respectively,
	\begin{align}
		O_{\max}(\rho^R,\gamma^R,\gamma^A)=\alpha_{\tr[\Pi^A\gamma^A]}(\vec{r}^R,\vec{g}^R).
	\end{align}
\end{prop}
Note that this Proposition provides a closed-form expression for arbitrary dimensions of $A$ and not only for qubits. Just as in Thm.~\ref{thm:CoolingHeating}, coherences between different energy eigenstates are irrelevant. Of course, it is also possible to interpret heating as maximizing the overlap with the largest energy eigenstate, but this seems to be of less technological relevance.  It is straightforward to see (and discussed in the Appendix after the proof of Prop.~3) that for qubits, the two interpretations of cooling and heating coincide. For ease of presentation, we thus use the notation corresponding to the first interpretation from here on.

\section{State transformations from cooling and heating}
Let $\tilde{\beta}_{\max}(\rho,\gamma;\beta,E)$ and $\tilde{\beta}_{\min}(\rho,\gamma;\beta,E)$ be the extremal $\tilde{\beta}$ to which we can cool respectively heat a qubit with energy gap $E$ at background inverse temperature $\beta$ given the resource $(\rho,\gamma)$ and access to $\CTO$. According to Thm.~\ref{thm:CoolingHeating}, these quantities have a closed-form expression which the families of functionals
\begin{align}\label{eq:Monotones}
	C_\beta^E(\rho,\gamma):=&\tilde{\beta}_{\max}(\rho,\gamma;\beta,E)-\beta, \nonumber \\
	H_\beta^E(\rho,\gamma):=&\beta-\tilde{\beta}_{\min}(\rho,\gamma;\beta,E),
\end{align}
inherit together with the natural interpretation that they capture the ability of $(\rho,\gamma)$ to cool and heat qubits. In addition, they are resource monotones, i.e., monotonic under $\CTO$, in the sense that $(\rho^R,\gamma^R) \xrightarrow{\CTO} (\sigma^S,\gamma^S)$ implies $C_\beta^E(\rho^R,\gamma^R)\ge C_\beta^E(\sigma^S,\gamma^S)$, see above the proof of Thm.~4 in the Appendix for a discussion and further properties. 
These monotones completely characterize conversions between quasi-classical states. 

\begin{thm}\label{thm:main}
	Let $(\sigma^S,\gamma^S)$ be quasi-classical. The following  statements are equivalent
	\begin{enumerate}
		\item $(\rho^R,\gamma^R) \xrightarrow{\CTO} (\sigma^S,\gamma^S)$.
		\item For any fixed $\beta>0$ and for all $E \in(0,\infty)$, it holds that
		\begin{align} \label{eq:thmConv1}
			&C_\beta^E(\rho^R,\gamma^R)\ge C_\beta^E(\sigma^S,\gamma^S),
			\nonumber \\
			&H_\beta^E(\rho^R,\gamma^R)\ge H_\beta^E(\sigma^S,\gamma^S).
		\end{align}
	\end{enumerate}
\end{thm}
The above Theorem shows that the performance in two of the most elementary thermodynamic tasks, namely cooling and heating of the simplest systems, i.e., qubits, fully determines the convertibility between and thus the relative value of quasi-classical athermality states. As we show in the Appendix after the proof, the conditions in 2 can be relaxed further: There exists a set of $|S|-1$ energies depending on the target state $(\sigma^S,\gamma^S)$ that correspond to \textit{either} a cooling or heating condition as in Eqs.~\eqref{eq:thmConv1} that in combination imply $(\rho^R,\gamma^R) \xrightarrow{\CTO} (\sigma^S,\gamma^S)$. 
In technical terms, this means that $C_\beta^E$ and $H_\beta^E$ form complete sets of monotones~\cite{Nielsen1999,Brandao2015,Gour2017,Gour2018,Rosset2018,Skrzypczyk2019,Takagi2019,Saxena2020,Gour2020a,Gour2021,Gour2020b,Datta2022} for quasi-classical thermodynamics under $\CTO$.

\section{Alternative characterization of state transformations}
So far, we answered the question of how much we can cool or heat specific qubits. However, in many applications, we want to reach a specific target inverse temperature $\tilde{\beta}$. A dual question of interest is therefore which qubits (characterized by their energy gap $E$) we can cool or heat to this $\tilde{\beta}$ given a resource $(\rho^R,\gamma^R)$ and background inverse temperature $\beta>0$~\cite{Janzing2000, Wilming2017}. In other words, we want to
find the set $\mathfrak{E}_{\beta}{(\rho^R, \gamma^R;\tilde{\beta})}$ of 
all energy gaps $E$ for which we can achieve the desired transformation. We start with an interesting observation.
\begin{prop}\label{prop:gap}
	For every $\beta,\tilde{\beta}>0, \beta\ne\tilde{\beta}$ there exist initial resources $(\rho^R,\gamma^R)$ such that $\mathfrak{E}_{\beta}{(\rho^R, \gamma^R;\tilde{\beta})}$ is \emph{not an interval}, i.e., there exist $E_1<E_2<E_3$ such that $E_1,E_3\in\mathfrak{E}_{\beta}{(\rho^R, \gamma^R;\tilde{\beta})}$ and $E_2\notin \mathfrak{E}_{\beta}{(\rho^R, \gamma^R;\tilde{\beta})}$.
\end{prop}
In the proof in the Appendix, we provide an explicit construction of such initial resource states for all finite $|R|>1$. 
The above Proposition may seem surprising, since naively, one might assume that if one can, e.g., heat a qubit with energy gap $E_3$ to $\tilde{\beta}$, one should be able to heat a qubit with energy gap $E_2<E_3$ to $\tilde{\beta}$ too. This is however not the case, see the Appendix for the proof and a discussion.
Nevertheless, the ability to cool and heat qubits to fixed temperatures again fully characterizes the conversions between quasi-classical states.
\begin{thm}\label{thm:alternative}
	Let $(\sigma^S, \gamma^S)$ be quasi-classical. 
	Then $(\rho^R, \gamma^R) \xrightarrow{\CTO} (\sigma^S, \gamma^S)$ if and only if 
	\begin{align}
		\mathfrak{E}_{\beta}{(\rho^R, \gamma^R;\tilde{\beta})} \supseteq \mathfrak{E}_{\beta}{(\sigma^S, \gamma^S;\tilde{\beta})}
	\end{align}
	 for any fixed $\beta>0$ and all $\tilde{\beta}\in(-\infty,\infty)$. 
\end{thm}
Whilst the $\mathfrak{E}_{\beta}$ represent sets and are therefore no resource monotones in the strict sense, the equation $\mathfrak{E}_{\beta}{(\rho^R, \gamma^R;\tilde{\beta})} \supseteq \mathfrak{E}_{\beta}{(\sigma^S, \gamma^S;\tilde{\beta})}$ can be seen as a generalized monotone~\cite{Gonda2019, Scandolo2021, Cockett2022} that implies an operational order, namely that $(\rho^R, \gamma^R)$ can be used to cool/heat more qubits to the desired target inverse temperature $\tilde{\beta}$ than $(\sigma^S, \gamma^S)$.

\section{Discussion}
In this work, we solved the problem to which temperatures we can cool and heat a system that is initially at equilibrium given access to closed thermal operations ($\CTO$) and an athermality state $(\rho,\gamma)$. Since $\CTO$ allows for the ultimate level of control, this provides bounds on the usefulness of thermal non-equilibrium when one is, e.g., interested in initializing a (qu)bit for (quantum) computation or other nanotechnological applications. Our results are also bounds in the sense mentioned in the introduction: If a system is initially not in its equilibrium state, we can first thermalize it and then cool or heat it using the given resource. Since thermalization is free, this cannot outperform optimal protocols that potentially do not include thermalization. Moreover, if one intends to cool a system initially not in its equilibrium state, one could in principle first use it to cool another system. During this process, it evolves towards its thermal state, and one can then cool it from there. Our results can thus, e.g., be used to determine how thermal non-equilibrium that is present at the end of a computation can be utilized to reduce the cost of re-initialization. At this point, we want to emphasize again that all of these results should be seen as fundamental bounds (similar to Landauer's principle~\cite{Landauer1961}), whilst in many contemporary experiments, one is orders of magnitudes away from these bounds. Nevertheless, it is interesting to determine what is possible in principle and thermal machines in the quantum limit aiming to approach these bounds are under active investigation~\cite{Scovil1959, Faucheux1995, Scully2002, Baugh2005, Uzdin2015, Silva2016, Klatzow2019, Bera2021}. For the case of qubits, we then used these temperatures to define resource monotones with a clear operational interpretation. These monotones are a complete set~\cite{Nielsen1999, Brandao2015, Gour2017, Gour2018, Rosset2018, Skrzypczyk2019, Takagi2019, Saxena2020, Gour2020a, Gour2021, Gour2020b, Datta2022} in the sense that they fully determine if quasi-classical states can be interconverted via $\CTO$. We thus characterized resource conversions, which are at the core of any resource theory, by the cooling and heating of qubits, i.e., by two of the most notable thermodynamical tasks on the simplest systems. 

Our work is related to the ``incoherent control paradigm'' of thermal machines discussed in Refs.~\cite{Clivaz2019, Clivaz2019b}, where also optimal cooling is investigated. The difference is that Refs.~\cite{Clivaz2019, Clivaz2019b} do not allow energy-preserving unitaries to act directly on parts of an arbitrary bath at fixed temperature, the target system, and a fixed system out of thermal equilibrium. Instead, they consider energy-preserving unitaries acting on the target system and a fixed thermal machine at thermal equilibrium (taking the role of our bath) as well as what they call an extension of their machine which is in a Gibbs state with respect to a different temperature (taking the role of our resource state). In this sense, their setting is different in that they do not allow for arbitrary baths and only consider resource states in Gibbs states at a different temperature. The ability to cool also plays a role in a resource theory of passivity~\cite{Singh2021}: a passive state that is virtually cooler than another one (see Refs.~\cite{Brunner2012, Skrzypczyk2015, Sparaciari2017} for the notion of virtual temperatures) implies that the former can cool a qubit better than the latter. As explored in Ref.~\cite{Lipka-Bartosik2023}, notions of cooling and heating can even be used to assign (two) temperatures to non-equilibrium states.

Technically, in our Theorems characterizing the conversion of athermality, we only required the target state to be quasi-classical. Since on quasi-classical states Gibbs-preserving operations ($\GPO$)~\cite{Faist2015,Egloff2015} (also studied under the name of asymmetric distinguishability~\cite{Matsumoto2010,Wang2019a,Wang2019b,Rethinasamy2020}) have the same conversion power as $\CTO$~\cite{Janzing2000,Shiraishi2020,Gour2022}, our results hold for $\GPO$ too if also the initial athermality state is quasi-classical. 
An interesting open question is if it is possible to characterize the convertibility between arbitrary athermality states (i.e., not quasi-classical ones) with the help of resource measures with an operational meaning (see Refs.~\cite{Buscemi2017, Narasimhachar2015} for some results concerning specific states). It follows from the discussion below Thm.~\ref{thm:CoolingHeating} that the monotones that we consider here are insufficient for this task, in essence because they only capture the non-uniformity aspect of athermality~\cite{Gour2015} but no restrictions emerging from coherence~\cite{Lostaglio2015} or, more precisely, covariance~\cite{Gour2008}. If it were possible to find (operational) monotones capable of describing the covariance aspects completely, in combination with our results, this might also shed light on the question in what sense $\CTO$ encapsulates more restrictions than Gibbs-preservation and covariance~\cite{Cwiklinski2015, Gour2018, Lostaglio2019, Gour2022, Ding2021}. 
An improved understanding of the operational capabilities of and differences between resource theories of quantum thermodynamics will in turn improve the understanding of, e.g., what level of control is needed for tasks such as cooling, and therefore lead to an improved design of (quantum) devices.

\begin{acknowledgments}
	We thank Frederik vom Ende for comments on the manuscript. 
	T. T., E. Z., and G. G. acknowledge support from the Natural Sciences and Engineering Research Council of Canada (NSERC). 
	T. T. acknowledges support from the Pacific Institute for the Mathematical Sciences (PIMS). The research and findings may not reflect those of the Institute.
	C. M. S. acknowledges the support of the Natural Sciences and Engineering Research Council of Canada (NSERC) through the Discovery Grant “The power of quantum resources” RGPIN-2022-03025 and the Discovery Launch Supplement DGECR-2022-00119.
\end{acknowledgments}

\appendix
\section{Additional notation and remarks}\label{app:NotRem}
In the main text, we introduced the notation $(\rho^A,\gamma^A)$ for athermality states to make it clear that we consider the resourcefulness of $\rho^A$ with respect to the reference Gibbs state $\gamma^A$. Since for a fixed inverse temperature $\beta$ the Gibbs state of a system is determined by its Hamiltonian, also the notation $(\rho^A,H^A)$ is found in the literature. 
We choose however to utilize Gibbs states instead of Hamiltonians, because the former determine the reference unambiguously, whilst the latter do not:  global energy shifts in a Hamiltonian, e.g., leave the Gibbs state unchanged, i.e., a state can be free with respect to many Hamiltonians, but only with respect to one Gibbs state. We also mentioned in the main text that the closed thermal operations ($\CTO$) are the topological closure of the thermal operations ($\TO$). Formally, we define them as follows.
\begin{definition}\label{def:CTO}
	Let $A$ and $A'$ be two physical systems with Gibbs states $\gamma^A, \gamma^{A'}$. We call $(\mathcal{N}^{A'\leftarrow A}, \gamma^A) $ a closed thermal operation 
	if and only if there exists a sequence of thermal operations $\{(\mathcal{N}^{A'\leftarrow A}_n,\gamma^A)\}_{n\in \mathbb{N}}$ with  $(\mathcal{N}^{A'\leftarrow A}_n,\gamma^A) \in \TO(\gamma^{A'}\leftarrow\gamma^A)\  \forall n \in \mathbb{N}$
	such that
	\begin{align}
		\lim_{n\rightarrow\infty}\norm{\mathcal{N}^{A'\leftarrow A}_n-\mathcal{N}^{A'\leftarrow A}}_\diamond=0,
	\end{align}
	where $\norm{\cdot}_\diamond=0$ denotes the diamond norm~\cite{Kitaev1997}.
	In analogy to $\TO(\gamma^{A'}\leftarrow\gamma^A)$, we write $\CTO(\gamma^{A'}\leftarrow\gamma^A)$.	
\end{definition}

Another potential choice of free operations is the set of Gibbs-preserving operations $(\GPO)$, which is the largest set of channels that preserves the free states, i.e., the Gibbs state. Analogously to, e.g., the  thermal operations, one denotes by $\mathrm{GPO}(\gamma^{A'}\leftarrow\gamma^{A})$ the set of Gibbs-preserving operations from system $A$ with Gibbs state $\gamma^A$ to system ${A'}$ with Gibbs state $\gamma^{A'}$, which is defined as
\begin{align}
	\left\{\mathcal{N}\in \mathrm{CPTP}\left(A'\leftarrow A\right): \mathcal{N}^{A'\leftarrow A}\left(\gamma^A\right)=\gamma^{A'} \right\}.
\end{align}
Gibbs-preserving operations will play an important role in our proofs.
We saw in the main text that thermal operations preserve the Gibbs state too, and therefore
\begin{align}\label{eq:TOinGPO}
	\mathrm{TO}\left(\gamma^{A'}\leftarrow\gamma^{A}\right)\subseteq \mathrm{GPO}\left(\gamma^{A'}\leftarrow\gamma^{A}\right).
\end{align}
We also discussed that we write $(\mathcal{E}^{A'\leftarrow A},\gamma^A)$ for a thermal operation since it is thermal with respect to the initial and target Gibbs state, but the target Gibbs state is determined by the operation and the initial Gibbs state via
\begin{align}\label{eq:TOpresGibbsSM}
	\mathcal{E}^{{A'}\leftarrow A}(\gamma^A)=&\tr_{B'}\left[U^{AB}\left(\gamma^A\otimes\gamma^B\right)U^{\dagger AB}\right] =\gamma^{A'}.
\end{align}
We use the same notation for closed thermal and Gibbs-preserving operations because we can also omit the target Gibbs states in these cases. For Gibbs-preserving operations this is clear by definition, and for closed thermal operations it follows from their definition in terms of thermal operations, i.e.,
\begin{align}
	\gamma^{A'}=\mathcal{N}^{A'\leftarrow A}\left( \gamma^A \right) =\lim_{n\rightarrow \infty} \mathcal{N}_n^{A'\leftarrow A}\left( \gamma^A_n \right).
\end{align} 
This also implies that 
\begin{align}\label{eq:CTOinGPO}
	\mathrm{CTO}\left(\gamma^{A'}\leftarrow\gamma^{A}\right)\subseteq \mathrm{GPO}\left(\gamma^{A'}\leftarrow\gamma^{A}\right).
\end{align} 

In Ref.~\cite[Lem.~II.1]{Gour2022}, it was shown that the requirement $AB=A'B'$ in the definition of the thermal operations can be relaxed to $AB\cong A'B'$. From a physical perspective, this is interesting since it implies that $AB$ and $A'B'$ can correspond to completely different physical systems. We will use this in the following to simplify our proofs.

At this point, we also want to emphasize that resource theories of thermodynamics are not designed to, e.g., explain why or how thermalization takes place, but rather study what is possible under the considered constraints. This is comparable to, e.g., the resource theory of entanglement~\cite{Plenio2007, Horodecki2009}, where the restriction to local operations and classical communication is well motivated, but also an ingredient to the theory and not a result of it. For technological applications, this practical approach is well suited, as illustrated by the immense success of classical thermodynamics, which is a \textit{classical} resource theory: by restricting ourselves to the manipulation of macroscopic degrees of freedom, we obtain the laws of thermodynamics from statistical physics. 

\section{Properties of closed thermal operations}\label{app:PropCTO}
In this section, we show that the set of closed thermal operations as defined in Def.~\ref{def:CTO} is meaningful in the sense that it leads to a well-defined resource theory.

First, we note that the closed thermal operations inherit closedness under concatenation from the thermal operations, which is essential for a consistent resource theory (see also Ref.~\cite[Prop.~4]{Ende2022} for the case $A=B=C$). 
\begin{lem}\label{lem:CTOclosedConc}
	If	
	\begin{align}
		\left(\mathcal{N}^{B\leftarrow A}, \gamma^A\right)\in \mathrm{CTO}\left(\gamma^{B}\leftarrow\gamma^{A}\right), \nonumber \\
		\left(\mathcal{E}^{C\leftarrow B}, \gamma^B\right)\in \mathrm{CTO}\left(\gamma^{C}\leftarrow\gamma^{B}\right),
	\end{align}
	then 
	\begin{align}
		\left(\mathcal{E}^{C\leftarrow B}\circ\mathcal{N}^{B\leftarrow A}, \gamma^A\right)\in \mathrm{CTO}\left(\gamma^{C}\leftarrow\gamma^{A}\right).
	\end{align}
\end{lem}
\begin{proof}
	Let $\{(\mathcal{N}^{B\leftarrow A}_n,\gamma^A)\}_{n\in \mathbb{N}}$ and $\{(\mathcal{E}^{C\leftarrow B}_n,\gamma^B)\}_{n\in \mathbb{N}}$
	be sequences of thermal operations as in Def.~\ref{def:CTO} that have $(\mathcal{N}^{B\leftarrow A},\gamma^A)$  respectively $(\mathcal{E}^{C\leftarrow B},\gamma^B)$ as their limits. Since thermal operations are closed under concatenation (which is a direct consequence of their definition), it follows that 
	\begin{align}
		(\mathcal{M}^{C\leftarrow A}_n,\gamma^A):=(\mathcal{E}^{C\leftarrow B}_n\circ \mathcal{N}_n^{B\leftarrow A},\gamma^A)\in \TO(\gamma^C\leftarrow\gamma^A).
	\end{align}
	Moreover, 
	\begin{align}
		& \norm{\mathcal{M}^{C\leftarrow A}_n - \mathcal{E}^{C\leftarrow B}\circ\mathcal{N}^{B\leftarrow A}}_\diamond \nonumber \\
		= & \norm{\mathcal{E}^{C\leftarrow B}_n\circ \mathcal{N}_n^{B\leftarrow A} - \mathcal{E}^{C\leftarrow B}\circ\mathcal{N}^{B\leftarrow A}}_\diamond \nonumber \\
		\le & \norm{\mathcal{E}^{C\leftarrow B}_n-\mathcal{E}^{C\leftarrow B}}_\diamond	+ \norm{\mathcal{N}_n^{B\leftarrow A}-\mathcal{N}^{B\leftarrow A}}_\diamond,
	\end{align}
	i.e.,
	\begin{align}
		\lim_{n\rightarrow \infty} \norm{\mathcal{M}^{C\leftarrow A}_n - \mathcal{E}^{C\leftarrow B}\circ\mathcal{N}^{B\leftarrow A}}_\diamond =0,
	\end{align}
	which finishes the proof.
\end{proof}
As intended, closed thermal operations resolve the problem discussed above their Definition in the main text:

\begin{lem}\label{lem:CTOepsilon}
	The following two statements are equivalent.
	\begin{enumerate}
		\item The athermality state $(\rho^A,\gamma^A)$ can be converted to $(\sigma^{A'}, \gamma^{A'})$ by $\CTO$, i.e.,
		\begin{align}
			(\rho^A,\gamma^A) \xrightarrow{\CTO} (\sigma^{A'}, \gamma^{A'}).
		\end{align}
		\item For every $\epsilon>0$, there exist athermality states $(\tilde{\rho}^A,\gamma^A)$ and
		$(\tilde{\sigma}^{A'}, \gamma^{A'})$ $\epsilon$-close to 
		$(\rho^A,\gamma^A)$ and $(\sigma^{A'}, \gamma^{A'})$,  respectively, such that 
		\begin{align}
			(\tilde{\rho}^A,\gamma^A)\xrightarrow{\TO}(\tilde{\sigma}^{A'}, \gamma^{A'}).
		\end{align}
	\end{enumerate}
\end{lem}
\begin{proof}
	We begin by showing that 2 follows from 1.
	Let $(\mathcal{N}^{A'\leftarrow A},\gamma^A)$ be a closed thermal operation that converts $(\rho^A,\gamma^A)$ to  $(\sigma^{A'}, \gamma^{A'})$ and let $\{(\mathcal{N}^{A'\leftarrow A}_n,\gamma^A)\}_{n\in \mathbb{N}}$ be a sequence of thermal operations as in Def.~\ref{def:CTO} that has $(\mathcal{N}^{A'\leftarrow A},\gamma^A)$ as its limit. 
	
	Next, introduce
	\begin{align}
		\sigma_n^{A'}:=\mathcal{N}^{A'\leftarrow A}_n \left( \rho^A\right).
	\end{align}	
	By definition, the thermal operation $(\mathcal{N}^{A'\leftarrow A}_n,\gamma^A)$ then converts $(\rho^A,\gamma^A)$ to $(\sigma_n^{A'},\gamma^{A'})$ and 
	\begin{align}
		\sigma^{A'}=\lim_{n\rightarrow\infty} \sigma^{A'}_n.
	\end{align}
	This further implies that for every $\epsilon>0$, there exists an $m$ such that
	\begin{align}
		\frac{1}{2} \norm{\sigma^{A'}_m-\sigma^{A'}}_1\le \epsilon, 
	\end{align}
	which finishes the first direction by choosing
	\begin{align}
		(\tilde{\rho}^{A},\gamma^{A})&=(\rho^{A'},\gamma^{A}), \nonumber \\ 
		(\tilde{\sigma}^{A'},\gamma^{A'})&=(\sigma^{A'}_m,\gamma^{A'}).
	\end{align} 
	
	Now let $\left\{\epsilon_k\right\}_{k \in \mathbb{N}}$ be a sequence of non-negative numbers such that 
	\begin{align}
		\lim_{k\rightarrow \infty}\epsilon_k=0.
	\end{align}
	To show the reverse, assume that 2 holds. This implies that for every $k\in\mathbb{N}$, there exists a $(\tilde{\rho}_k^A, \gamma^A)$ $\epsilon_k$-close to $(\rho^A, \gamma^A)$ and a thermal operation $(\mathcal{N}^{A'\leftarrow A}_k,\gamma^A)$ that converts $(\tilde{\rho}_k^A, \gamma^A)$ to a state $(\tilde{\sigma}_k^{A'},\gamma^{A'})$ $\epsilon_k$-close to $(\sigma^{A'},\gamma^{A'})$. Since the set of quantum channels from $A$ to $A'$ is compact, there exists a converging sub-sequence $\{\mathcal{N}^{A'\leftarrow A}_{k_l}\}_l$. We can thus define
	\begin{align}
		\mathcal{N}^{A'\leftarrow A}:=\lim_{l\to \infty} \mathcal{N}^{A'\leftarrow A}_{k_l}.
	\end{align}
	By definition, $(\mathcal{N}^{A'\leftarrow A},\gamma^A)\in\CTO(\gamma^{A'}\leftarrow\gamma^A)$. Moreover, 
	\begin{align}
		\mathcal{N}^{A'\leftarrow A}(\rho^A)=\lim_{l\to \infty} \mathcal{N}^{A'\leftarrow A}_{k_l}(\tilde{\rho}_{k_l}^A)= \lim_{l\to \infty}\tilde{\sigma}_{k_l}^{A'}= \sigma^{A'}, 
	\end{align}
	which finishes the proof.
\end{proof}
In summary, in the above Lemma, we showed that the introduction of $\CTO$ allows us to alter $\rho^A$ and $\sigma^{A'}$ arbitrarily little, which is desirable from a physics perspective. In the main text, we were mainly concerned about only being able to approximate the target state arbitrarily well. But the moment we intend to convert athermality states with concatenated operations, also allowing for an arbitrarily little variation in the input state is of interest. Similarly, one could also demand that one is allowed to alter the Gibbs states slightly and adapt Def.~\ref{def:CTO} of the closed thermal operations accordingly. This leads however to the problem that it is not clear if an analog of Lem.~\ref{lem:CTOclosedConc} exists in this case, which is needed to ensure that the resulting resource theory is well-defined (in the usual framework of resource theories). 

Another convenient property of a resource theory is convexity since it allows us to use many mathematical tools and methods from convex analysis~\cite{Chitambar2019}. Indeed, the closed thermal operations are convex, see Ref.~\cite[App.~C]{Lostaglio2015}, Ref.~\cite[Prop.~4]{Ende2022}, and Ref.~\cite[Thm.~II.1]{Gour2022}.

For the task of cooling and heating that we consider in this work, we are specifically interested in quasi-classical target states, for which the following holds (see also Ref.~\cite[Thm.~ 6]{Janzing2000} and Ref.~\cite{Lipka-Bartosik2023})
\begin{lem}\label{lem:targetQuasiClassical}
	Assume $(\sigma^A,\gamma^A)$ is quasi-classical. Then
	\begin{align}
		(\rho^B,\gamma^B)&\xrightarrow{\CTO} (\sigma^A,\gamma^A) \nonumber \\
		\Leftrightarrow (\mathcal{P}_{\gamma^B}(\rho^B),\gamma^B) &\xrightarrow{\CTO}(\sigma^A,\gamma^A).
	\end{align}
\end{lem}
\begin{proof}
	We first assume that $(\mathcal{P}_{\gamma^B}(\rho^B),\gamma^B) \xrightarrow{\CTO}(\sigma^A,\gamma^A)$ holds, i.e., there exists an $\mathcal{E}\in\CTO(\gamma^B\leftarrow\gamma^A)$ such that 
	\begin{align}
		\mathcal{E}(\mathcal{P}_{\gamma^B}(\rho^B))=\sigma^A.
	\end{align}
	Due to Lem.~\ref{lem:CTOclosedConc} and $\mathcal{P}_{\gamma^B}(\gamma^B)=\gamma^B$, 
	$\mathcal{G}:=\mathcal{E}\circ\mathcal{P}_{\gamma^B}\in\CTO(\gamma^B\leftarrow\gamma^A)$. Moreover, 
	\begin{align}
		\mathcal{G}(\rho^B)=\sigma^A,
	\end{align}
	i.e., $(\rho^B,\gamma^B)\xrightarrow{\CTO} (\sigma^A,\gamma^A)$.

	For the reverse, assume that $(\rho^B,\gamma^B)\xrightarrow{\CTO} (\sigma^A,\gamma^A)$, i.e., that there exists a sequence of thermal operations $\mathcal{E}_n\in\TO(\gamma^A\leftarrow\gamma^B)$ such that 
	\begin{align}
		\mathcal{E}(\rho^B):=\lim_{n\rightarrow\infty} \mathcal{E}_n(\rho^B)=\sigma^A.
	\end{align}	
	Now every thermal operation is covariant in the sense that~\cite{Lostaglio2015b}
	\begin{align}
		\mathcal{P}_{\gamma^A}\circ\mathcal{E}_n=\mathcal{E}_n\circ \mathcal{P}_{\gamma^B},
	\end{align}
	and since the set of covariant operations is closed, also $\mathcal{E}$ is covariant. Therefore
	\begin{align}\label{eq:offdiagCool}
		\sigma^A=&\mathcal{P}_{\gamma^A}(\sigma^A)= \mathcal{P}_{\gamma^A} \circ  \mathcal{E}(\rho^B) = \mathcal{E}\circ \mathcal{P}_{\gamma^B}(\rho^B), 
	\end{align}
	and $(\mathcal{P}_{\gamma^B}(\rho^B),\gamma^B) \xrightarrow{\CTO}(\sigma^A,\gamma^A)$.
\end{proof}

\section{Proofs of the results in the main text}\label{app:Proofs}

In the following, we provide the proofs of the results presented in the main text.
Since all of them are related to the cooling or heating of a system $A$, it is useful to expand the Hamiltonian of $A$ explicitly as
\begin{align}
	H^A=\sum_{i=1}^{|A|} h_i\ketbra{i}{i}:\quad h_1\le h_2\le...\le h_{|A|}.
\end{align}

We start with some general observations concerning cooling and heating according to the first interpretation mentioned in the main text, i.e., the creation of 
\begin{align}
	(\tilde{\gamma}^A,\gamma^A)
\end{align}
given an initial resource $(\rho^R,\gamma^R)$ and access to $\CTO$.
This is possible iff 
\begin{align}
	(\rho^R,\gamma^R) &\xrightarrow{\CTO} (\tilde{\gamma}^A,\gamma^A),
\end{align}
which is according to Lem.~\ref{lem:targetQuasiClassical} equivalent to
\begin{align}
	(\mathcal{P}_{\gamma^R}(\rho^R),\gamma^R) &\xrightarrow{\CTO} (\tilde{\gamma}^A,\gamma^A),
\end{align}
i.e., the problem can be reduced to conversions between quasi-classical states and therefore probability vectors. Importantly, Ref.~\cite{Janzing2000} (see also Refs.~\cite{Horodecki2013, Shiraishi2020}) showed that between quasi-classical states, the conversion power of $\CTO$ is equal to the conversion power of $\GPO$. 
From here on, we denote by $\vec{r}^R, \vec{g}^R, \tilde{\vec{g}}^A,\vec{g}^A$ the probability vectors corresponding to $\mathcal{P}_{\gamma^R}(\rho^R), \gamma^R,\tilde{\gamma}^A,\gamma^A$, respectively, i.e.,
\begin{align}
	\gamma^A=\sum_i g_i^A \ketbra{i}{i}^A:\quad g_i^A=\frac{e^{-\beta h_i}}{Z(\beta)}, \nonumber \\
	\tilde{\gamma}^A=\sum_i \tilde{g}_i^A \ketbra{i}{i}^A:\quad \tilde{g}_i^A=\frac{e^{-\tilde{\beta}h_i}}{Z(\tilde{\beta})}.
\end{align}	
According to Eq.~\eqref{eq:relMajo} in the main text and the following discussion, the question if we can use $(\rho^R,\gamma^R)$ to cool/heat system $A$ to a specific $\tilde{\gamma}^A$ is therefore equivalent to the question if
\begin{align}
	\alpha_y(\vec{r}^R,\vec{g}^R)\ge\alpha_y(\tilde{\vec{g}}^A,\vec{g}^A) \ \forall y\in[0,1].
\end{align}
Due to convexity, and as mentioned in Ref.~\cite{Renes2016}, this is in turn equivalent to
\begin{align}\label{eq:conditionsTempChange}
	\alpha_{y_k^\star(\tilde{\vec{g}}^A,\vec{g}^A)}(\vec{r}^R,\vec{g}^R)\ge x_k^\star(\tilde{\vec{g}}^A,\vec{g}^A)\ \forall k\in[|A|],
\end{align}
where $x_k^\star$ and $y_k^\star$ were defined in Eq.~\eqref{eq:elbows} in the main text.

From here on we treat heating and cooling separately, starting with the latter:  
For $\tilde{\beta}\ge \beta>0$
\begin{align}\label{eq:orderingGibbs}
	\frac{\tilde{g}_i^A}{g_i^A}=\frac{Z(\beta)}{Z(\tilde{\beta})}  e^{-h_i(\tilde{\beta}-\beta)}
\end{align}
is already ordered non-increasingly due to the assumed order of the $h_i$ and
\begin{align}\label{eq:elbowsCooling}
	(x^\star_k(\tilde{\vec{g}}^A,\vec{g}^A),y^\star_k (\tilde{\vec{g}}^A,\vec{g}^A))=&\left(\sum_{j=1}^{k}\tilde{g}^A_{j},\sum_{j=1}^{k}g_{j}^A\right)
	=\left(\norm{\tilde{\vec{g}}^A}_{(k)},\norm{\vec{g}^A}_{(k)}\right),
\end{align}
where we used Ky-Fan norms as introduced in Eq.~\eqref{eq:KyFanNorm}. Importantly, according to the above equation, the $y^\star_k (\tilde{\vec{g}}^A,\vec{g}^A)$ only depend on the Gibbs state of the target system at the fixed background temperature and are thus independent of $\tilde{\beta}$ (the inverse temperature to which we attempt to cool our target system), but the $x^\star_k (\tilde{\vec{g}}^A,\vec{g}^A)$ increase with increasing $\tilde{\beta}$ because $\norm{\tilde{\vec{g}}^A}_{(k)}$ increases in that case.
Plugging Eq.~\eqref{eq:elbowsCooling} into Eq.~\eqref{eq:conditionsTempChange}, we find the conditions 
\begin{align}\label{eq:coolingCondition}
	\norm{\tilde{\vec{g}}^A}_{(k)}\le \alpha_{ \norm{\vec{g}^A}_{(k)}}(\vec{r}^R,\vec{g}^R)=:\alpha_k\ \forall k\in[|A|].
\end{align}
The right-hand sides of these inequalities are again independent of $\tilde{\beta}$ but depend on the initial resource $(\rho^R,\gamma^R)$ as well as the background inverse temperature $\beta$. Now recall that the function $y \mapsto\alpha_y(\vec{r}^R,\vec{g}^R)$ describes the $x$-values of the lower boundary of the testing region associated with the initial resource $(\rho^R,\gamma^R)$ as a function of the $y$-values. Also recall that this lower boundary is obtained by connecting the elbows $(x_k^\star(\vec{r}^R,\vec{g}^R),y_k^\star(\vec{r}^R,\vec{g}^R))$ with straight lines, see Fig.~\ref{fig:VisualizationProof}. For $y_k^\star(\vec{r}^R,\vec{g}^R)\le y < y_{k+1}^\star(\vec{r}^R,\vec{g}^R)$, it thus holds that 
\begin{align} \nonumber
	\alpha_y(\vec{r}^R,\vec{g}^R)=x_{k}^\star(\vec{r}^R,\vec{g}^R)+\frac{x^\star_{k+1}(\vec{r}^R,\vec{g}^R)-x_{k}^\star(\vec{r}^R,\vec{g}^R)}{y^\star_{k+1}(\vec{r}^R,\vec{g}^R)-y_{k}^\star(\vec{r}^R,\vec{g}^R)} \left(y-y_{k}^\star(\vec{r}^R,\vec{g}^R)\right),
\end{align}
since this describes the straight line through the two points  $(x_k^\star(\vec{r}^R,\vec{g}^R),y_k^\star(\vec{r}^R,\vec{g}^R))$ and $(x_{k+1}^\star(\vec{r}^R,\vec{g}^R),y_{k+1}^\star(\vec{r}^R,\vec{g}^R))$. By choosing $l_k$ such that $y_{l_k}^\star(\vec{r}^R,\vec{g}^R)\le\norm{\vec{g}^A}_{(k)}< y^\star_{l_{k}+1}(\vec{r}^R,\vec{g}^R)$, it thus holds that
\begin{align}
	&\alpha_k= \alpha_{ \norm{\vec{g}^A}_{(k)}}(\vec{r}^R,\vec{g}^R)=x_{l_k}^\star(\vec{r}^R,\vec{g}^R)+\frac{x^\star_{l_{k}+1}(\vec{r}^R,\vec{g}^R)-x_{l_k}^\star(\vec{r}^R,\vec{g}^R)}{y^\star_{l_{k}+1}(\vec{r}^R,\vec{g}^R)-y_{l_k}^\star(\vec{r}^R,\vec{g}^R)} \left(\norm{\vec{g}^A}_{(k)}-y_{l_k}^\star(\vec{r}^R,\vec{g}^R)\right),
\end{align}
i.e., we expressed $\alpha_k$ in terms of $\norm{\vec{g}^A}_{(k)}$ and the elbows associated with the initial resource. Importantly, this shows that the $\alpha_k$ can be directly calculated once the initial resource and $\gamma^A$ are specified.

For qubits, one can now determine the maximal $\tilde{\beta}$ explicitly: Excluding the trivial case of a qubit with energy gap $E=0$, the only relevant condition is that $\tilde{g}_1^A\le\alpha_1$. Noting that $\tilde{g}_1^A$ increases with increasing $\tilde{\beta}$, it follows that the optimal $\tilde{g}_1^A$ is given by $\tilde{g}_1^A=\alpha_1$ and since
\begin{align}
	\tilde{g}_1^A=\frac{1}{1+e^{-E\tilde{\beta}}},
\end{align}
we conclude that
\begin{align}\label{eq:SMbetaMax}
	\tilde{\beta}_{\max}=\frac{1}{E}\ln\left(\frac{\alpha_1}{1-\alpha_1}\right).
\end{align}
\begin{figure}[tb!]%
	\centering
	\includegraphics[width=0.6\linewidth]{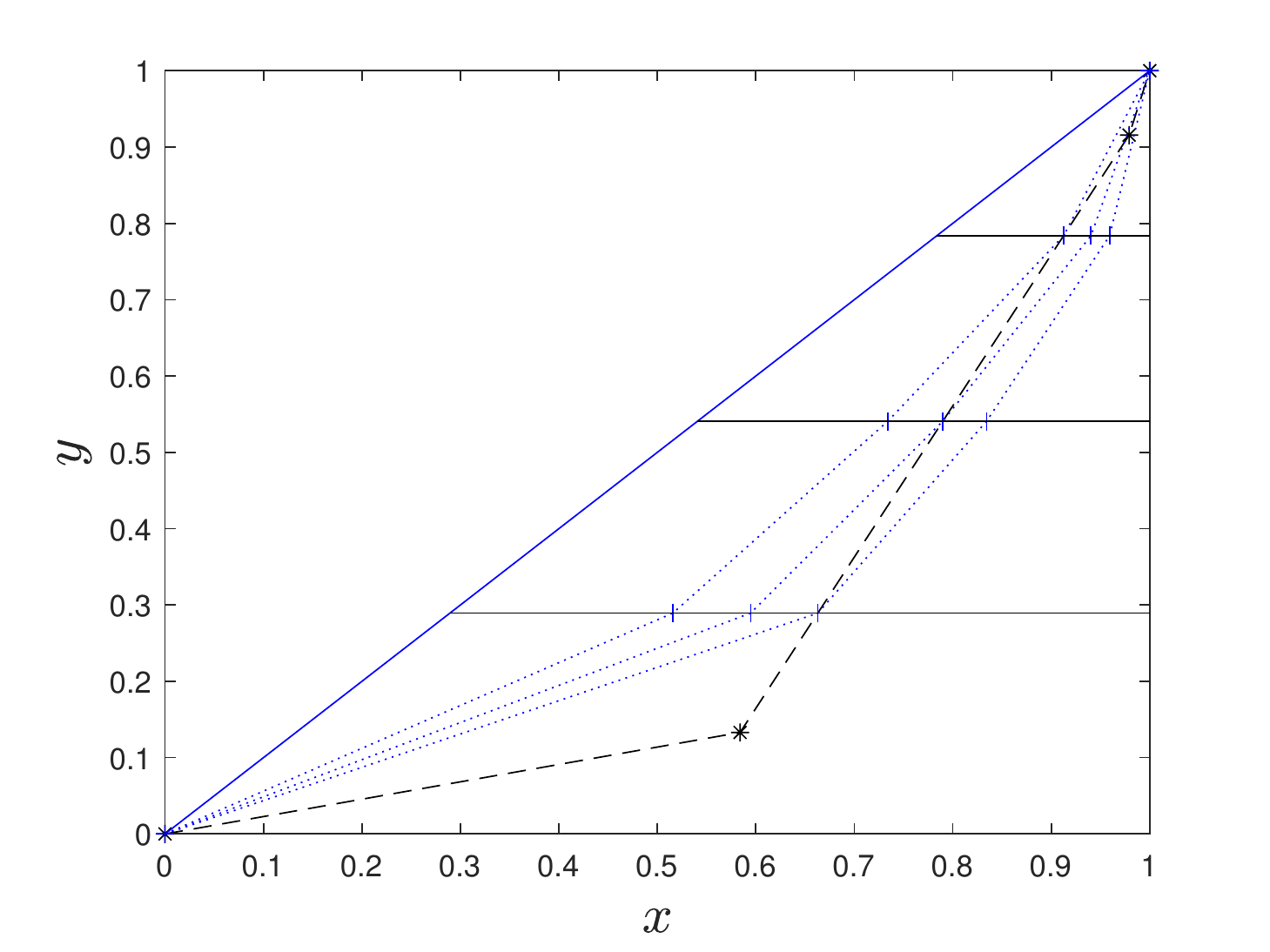}
	\caption{Visualization of the proof concerning the cooling of arbitrary systems. The dashed black line represents the lower boundary of the testing region associated with the initial resource which is fully characterized by its elbows represented by the black * symbols. If we do not cool the target system, its testing region is the straight line connecting the origin and $(1,1)$ (solid, blue). Once we start cooling, the elbows associated with the target system move along the solid black lines, since the $y^\star_k(\tilde{\vec{g}}^A,\vec{g}^A)$ are independent of $\tilde{\beta}$, but the $x^\star_k(\tilde{\vec{g}}^A,\vec{g}^A)$ are not, compare Eq.~\eqref{eq:elbowsCooling}. Importantly, the $x^\star_k(\tilde{\vec{g}}^A,\vec{g}^A)$ increase monotonically in $\tilde{\beta}$. The $\tilde{\beta}_k$ correspond to the $k$-th elbow sitting on the lower boundary associated with the initial resource, leading to different lower boundaries of the cooled target system (blue, dotted). These are the maximal $\tilde{\beta}$ that satisfy the $k$-th condition in Eq.~\eqref{eq:coolingCondition}. If $\beta\le\tilde{\beta}<\tilde{\beta}_k$, the $k$-th condition in Eq.~\eqref{eq:coolingCondition} is satisfied too (due to the monotonicity). Therefore $\tilde{\beta}_{\max}=\min_{k\in[n-1]} \tilde{\beta}_k$ is the largest $\tilde{\beta}$ to which we can cool the target system given the initial resource because it is the largest $\tilde{\beta}$ such that the lower boundary of the cooled target system is fully above the lower boundary of the initial resource.
	} \label{fig:VisualizationProof}
\end{figure}%

Excluding again the case of a completely degenerate Hamiltonian, in the general case, for $1\le k<n:=|A|$, let $\tilde{\beta}_k\ge \beta$ be such that
\begin{align}
	\norm{\tilde{\vec{g}}^A(\tilde{\beta}_k)}_{(k)}=\alpha_k.
\end{align}
From the geometry of the testing region, it is apparent that such $\tilde{\beta}_k$ always exist and are unique (see Fig.~\ref{fig:VisualizationProof}). However, for general Hamiltonians, these equations need to be solved numerically (which is very simple due to monotonicity).
Then
\begin{align}
	\tilde{\beta}_{\max}=\min_{k\in[n-1]} \tilde{\beta}_k.
\end{align}
This can be seen by the following argument (which is essentially a geometric argument on the testing region, see Fig.~\ref{fig:VisualizationProof} for a visualization): First, we notice that for $k=n$, Eq.~\eqref{eq:coolingCondition} is always satisfied, so we only consider $1\le k<n$. It is straightforward to see that $||\tilde{\vec{g}}^A(\tilde{\beta})||_{(k)}$, the probability to find the system in one of its $k$ lowest energy eigenstates, increases iff $\tilde{\beta}$ increases. The $\tilde{\beta}_k$ are thus the highest inverse temperatures that satisfy the $k$-th condition. Since each individual condition is also satisfied for all inverse temperatures between $\tilde{\beta}_k$ and $\beta$, we arrive at Eq.~\eqref{eq:minT} in the main text.

Moving on to the heating case, we note that for $\tilde{\beta}\le \beta, \beta>0$, $\tilde{g}_i/g_i$ is ordered non-decreasingly (see Eq.~\eqref{eq:orderingGibbs}), and therefore
\begin{align}\label{eq:elbowsHeating}
	(x^\star_k&(\tilde{\vec{g}}^A,\vec{g}^A),y^\star_k (\tilde{\vec{g}}^A,\vec{g}^A))=\left(\sum_{j=n-k+1}^{n}\tilde{g}_{j},\sum_{j=n-k+1}^{n}g_{j}\right)=\left(\sum_{j=n-k+1}^{n}\tilde{g}_{j},1-\norm{\vec{g}^A}_{(n-k)}\right),
\end{align}
where we understand that $\norm{\vec{g}}_{(0)}=0$.
Plugging into Eq.~\eqref{eq:conditionsTempChange}, the analogue of Eq.~\eqref{eq:coolingCondition} is therefore 
\begin{align}
	\sum_{j=n-k+1}^{n}\tilde{g}_{j}\le \alpha_{1- \norm{\vec{g}^A}_{(n-k)}}(\vec{r}^R,\vec{g}^R)=:\tilde{\alpha}_{k}\ \forall k\in [n], 
\end{align}
where, with $\tilde{l}_k$ such that $y_{\tilde{l}_k}^\star(\vec{r}^R,\vec{g}^R)\le1-\norm{\vec{g}^A}_{(n-k)}< y^\star_{\tilde{l}_{k}+1}(\vec{r}^R,\vec{g}^R)$,
\begin{align}
	\tilde{\alpha}_k=&x_{\tilde{l}_k}^\star(\vec{r}^R,\vec{g}^R) +\frac{x^\star_{\tilde{l}_{k}+1}(\vec{r}^R,\vec{g}^R)-x_{\tilde{l}_k}^\star(\vec{r}^R,\vec{g}^R)}{y^\star_{\tilde{l}_{k}+1}(\vec{r}^R,\vec{g}^R)-y_{\tilde{l}_k}^\star(\vec{r}^R,\vec{g}^R)} \left(1-\norm{\vec{g}^A}_{(n-k)}-y_{\tilde{l}_k}^\star(\vec{r}^R,\vec{g}^R)\right).
\end{align}

Considering the qubit-heating case, we notice that the only non-trivial condition is
\begin{align}
	\tilde{g}_2^A\le\tilde{\alpha}_1.
\end{align}
This time, lower $\tilde{\beta}$ means larger $\tilde{g}_2^A$, so we choose $\tilde{g}_2^A=\tilde{\alpha}_1$. Here, in contrast to the cooling case, we need to consider the following: since $\tilde{g}_2^A$ is the population of the qubit's higher energy level, it cannot be larger than $1/2$ unless we allow for negative temperatures associated with population inversions. We will see later that this is indeed useful, and therefore include negative $\tilde{\beta}$. 
The lowest inverse temperature $\tilde{\beta}$ that we can reach is given by 
\begin{align}\label{eq:TmaxQubit}
	\tilde{\beta}_{\min} =\frac{1}{E} \ln\left(\frac{1-\tilde{\alpha}_1}{\tilde{\alpha}_1}\right).
\end{align}	
Allowing for negative temperatures is also the reason why in Eq.~\eqref{eq:elbowsHeating}, we cannot replace  $\sum_{j=n-k+1}^{n}\tilde{g}_{j}$ with $1-||\tilde{\vec{g}}^A(\beta_k)||_{(n-k)}$.

The heating of general systems can be treated in analogy to the cooling case too: Let $\tilde{\beta}_k$ be such that 
\begin{align}
	\sum_{j=n-k+1}^{n}\tilde{g}_{j}(\tilde{\beta}_k)=\tilde{\alpha}_k.
\end{align}
The minimal inverse temperature to which we can heat is then given by
\begin{align}
	\tilde{\beta}_{\min}=\max_{k\in[n-1]} \tilde{\beta}_k.
\end{align}
In Fig.~\ref{fig:Qubit} and Fig.~\ref{fig:HigherDim}, we plot the lower boundaries of the testing regions corresponding to these results.
\begin{figure}[tb!]%
	\centering
	\includegraphics[width=0.5\linewidth]{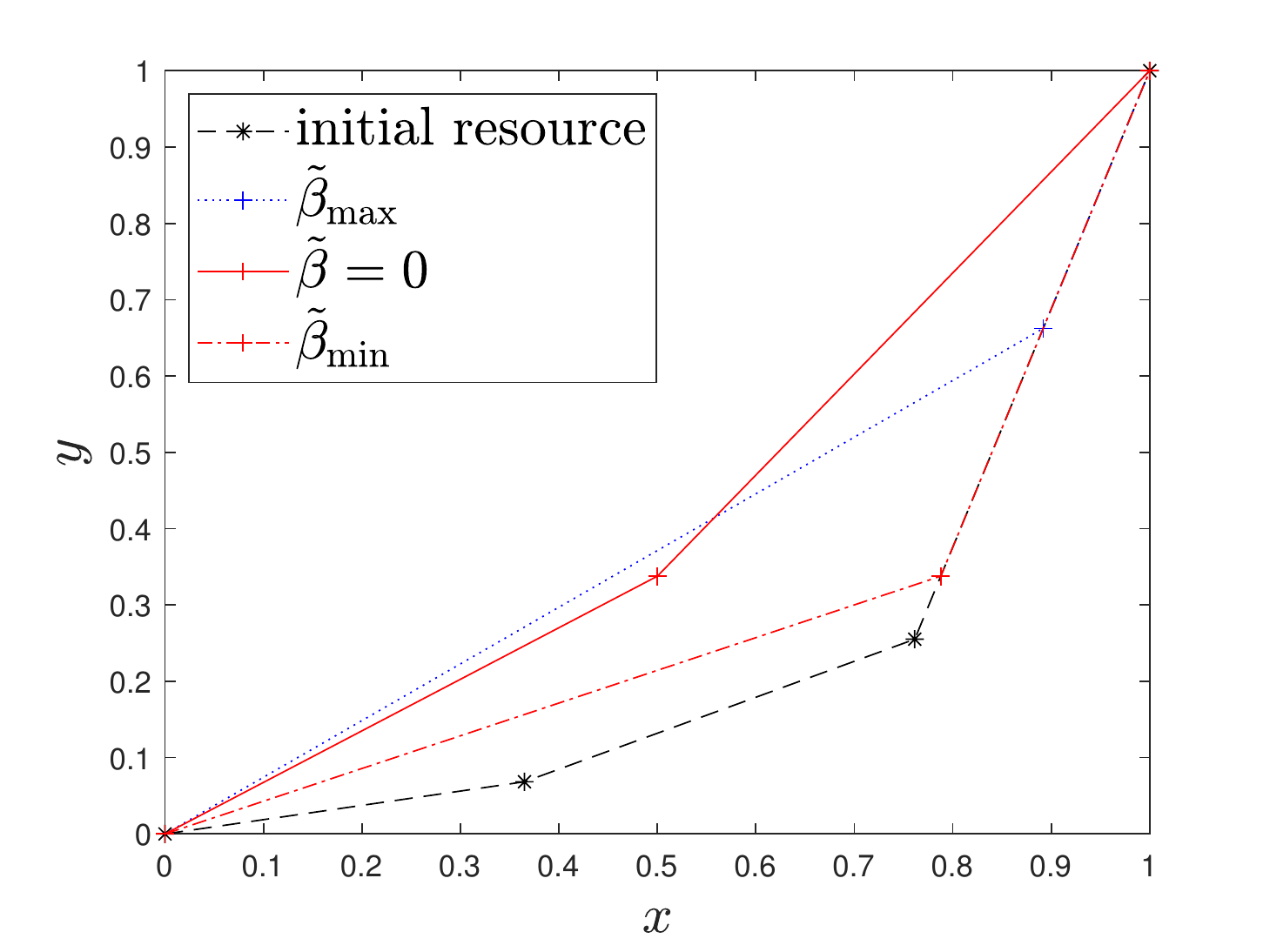}
	\caption{Lower boundaries of testing regions relevant for the cooling and heating of a qubit. For a given initial athermality state (black, dashed) with two non-trivial corresponding elbows, i.e., of dimension three, the minimal temperature ($\tilde{\beta}_{\max}$) of the qubit that can be reached corresponds to the blue dotted line. It is possible to heat the qubit to infinite temperature ($\tilde{\beta}=0$, red, solid), and the lowest reachable $\tilde{\beta}$, i.e., $\tilde{\beta}_{\min}$ corresponds to a population inversion (red, dashed-dotted).} \label{fig:Qubit}
\end{figure}%
\begin{figure}[tb!]%
	\centering
	\includegraphics[width=0.5\linewidth]{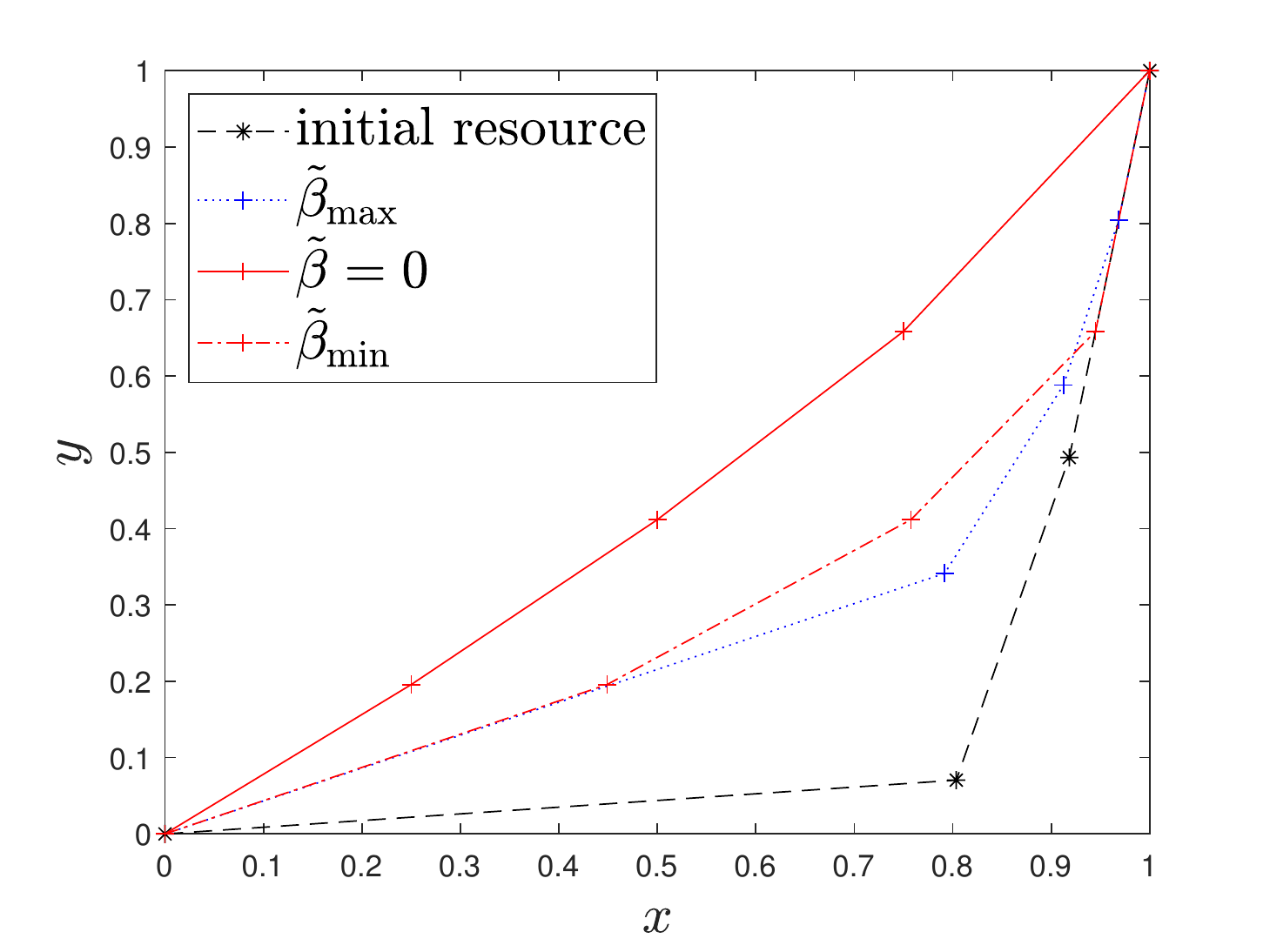}
	\caption{Lower boundaries of testing regions relevant for the cooling and heating of a system of dimension four. For a given initial athermality state (black, dashed) of dimension three, the minimal temperature ($\tilde{\beta}_{\max}$) of the target system that can be reached corresponds to the blue dotted line. It is possible to heat the target system to infinite temperature ($\tilde{\beta}=0$, red, solid), and the lowest reachable $\tilde{\beta}$, i.e., $\tilde{\beta}_{\min}$ corresponds to a population inversion (red, dashed-dotted).}\label{fig:HigherDim}
\end{figure}%
In summary, we have proven Thm.~\ref{thm:CoolingHeating} from the main text and on the way established some basic results which we also use in the following proofs, including the proof of Prop.~\ref{prop:altCooling}, which we restate for readability.

\setcounter{theorems}{2}
\begin{prop}
	With $\vec{r}^R, \vec{g}^R,\vec{g}^A$ denoting the probability vectors corresponding to $\mathcal{P}_{\gamma^R}(\rho^R), \gamma^R,\gamma^A$, respectively,
	\begin{align}\label{eq:Omax}
		O_{\max}(\rho^R,\gamma^R,\gamma^A)=\alpha_{\tr[\Pi^A\gamma^A]}(\vec{r}^R,\vec{g}^R).
	\end{align}
\end{prop}

\begin{proof}
	Remember the definition of $O_{\max}$, i.e.,
	\begin{align}
		O_{\max}&(\rho^R,\gamma^R,\gamma^A)	=\max\{\tr\left[\Pi^A\tau^A\right]: (\rho^R,\gamma^R)\xrightarrow{\CTO} (\tau^A,\gamma^A) \},
	\end{align}
	and let $\tau_{\star}^A$ be an optimizer of this optimization problem.
	Since $\tr\left[\Pi^A\tau^A_{\star}\right]=\tr\left[\Pi^A\mathcal{P}_{\gamma^A}(\tau^A_{\star})\right]$, we can restrict ourselves to quasi-classical target states and therefore assume without loss of generality that $(\rho^R,\gamma^R)$ is quasi-classical too (see Lem.~\ref{lem:targetQuasiClassical}). Assuming that the ground state has degeneracy $d$, this implies that there exists an orthonormal basis $\{\ket{i}\}$ of $A$ such that
	\begin{align}
		\tau_{\star}^A=&\sum_{i=1}^A t_i^A\ketbra{i}{i},  \\
		H^A=&\sum_{i=1}^A h_i^A\ketbra{i}{i}:\ h_1=...=h_d<h_{d+1}\le...\le h_{|A|}\nonumber
	\end{align}
	and thus $O_{\max}(\rho^R,\gamma^R,\gamma^A)=\sum_{i=1}^dt_i^A$. Next, we notice that we can further restrict ourselves to the case $t_1^A=t_2^A=...=t_d^A$: Any unitary acting non-trivially only on the support of $\Pi^A$ is a thermal operation. By applying such unitaries uniformly at random (which is in CTO, since CTO is convex, see Ref.~\cite[App.~C]{Lostaglio2015}, Ref.~\cite[Prop.~4]{Ende2022}, and Ref.~\cite[Thm.~II.1]{Gour2022}), for $i\in{1,...,d}$, we map $t_i^A$ to $\sum_{i=1}^dt_i^A/d$. Obviously, this does not change $O_{\max}$.
	
	For the moment, assume $O_{\max}<1$ and let 
	\begin{align}
		\vec{t}^{\tilde{A}}=\frac{1}{1-dt_1^A}
		\begin{pmatrix}
			t_{d+1}^A \\
			t_{d+2}^A \\
			\vdots \\
			t_{|A|}^A
		\end{pmatrix}, \ \vec{g}^{\tilde{A}}=\frac{1}{1-dg_1^A}
		\begin{pmatrix}
			g_{d+1}^A \\
			g_{d+2}^A \\
			\vdots \\
			g_{|A|}^A
		\end{pmatrix},
	\end{align}
	which by construction are valid probability vectors. Since $(\vec{t}^{\tilde{A}},\vec{g}^{\tilde{A}})\succ(\vec{g}^{\tilde{A}},\vec{g}^{\tilde{A}})$, by definition, there exists a column stochastic matrix $\tilde{E}$ such that $\tilde{E}\vec{t}^{\tilde{A}}=\vec{g}^{\tilde{A}}$ and $\tilde{E}\vec{g}^{\tilde{A}}=\vec{g}^{\tilde{A}}$. Then
	\begin{align}
		E=\begin{pmatrix}
			I_d & 0 \\
			0 & \tilde{E}
		\end{pmatrix},
	\end{align}
	where $I_d$ is the identity matrix of dimension $d$, is column stochastic too, and, with $\lambda=\frac{1-dt_1^A}{1-dg_1^A}$,
	\begin{align}
		E\vec{g}^A=& E \begin{pmatrix}
			g_1^A \\
			\vdots \\
			g_d^A \\
			(1-dg_1^A)\vec{g}^{\tilde{A}}
		\end{pmatrix}=\vec{g}^A, \\	
		E\vec{t}^A=&E \begin{pmatrix}
			t_1^A \\
			\vdots\\
			t_d^A\\
			(1-dt_1^A)\vec{t}^{\tilde{A}}
		\end{pmatrix}=
		\begin{pmatrix}
			t_1^A \\
			\vdots\\
			t_d^A\\
			(1-dt_1^A)\vec{g}^{\tilde{A}} 
		\end{pmatrix}
		=
		\begin{pmatrix}
			t_1^A \\
			\vdots\\
			t_d^A\\
			\lambda g_{d+1}^A \\
			\lambda g_{d+2}^A \\
			\vdots \\
			\lambda g_{|A|}^A
		\end{pmatrix}=:\tilde{\vec{t}}^A(t_1^A),
	\end{align}
	where we used in the notation that $t_1^A=t_2^A=...=t_d^A$.
	This implies that $(\vec{t}^{A},\vec{g}^{A})\succ(\tilde{\vec{t}}^{A}(t_1^A),\vec{g}^{A})$. Moreover, $O_{\max}(\rho^R,\gamma^R,\gamma^A)=d\tilde{t}_1^A$, and thus we can and will restrict our considerations to targets of the form $\tilde{\vec{t}}^A(t_1^A)$, which also includes the special case $O_{\max}(\rho^R,\gamma^R,\gamma^A)=1$ again. 
	
	Since 
	\begin{align}
		(\vec{r}^R,\vec{g}^R)\succ (\vec{g}^A,\vec{g}^A),
	\end{align}
	we find  that $O_{\max}(\rho^R,\gamma^R,\gamma^A)=d\tilde{t}_1^A\ge dg_1^A$. We thus have that $\tilde{t}_i^A/g_i^A$ is ordered non-increasingly for an ideal $\tilde{\vec{t}}^A$, which implies that
	\begin{align}\label{eq:overlapMax}
		&\left(x_k^\star(\tilde{\vec{t}}^A,\vec{g}^A),y_k^\star(\tilde{\vec{t}}^A,\vec{g}^A)\right)=\left(\tilde{t}_1^A \min\{k,d\}+\lambda\sum_{j=d+1}^{k}g_j^A, \sum_{j=1}^{k}g_j^A\right).
	\end{align}
	
	Now we remember that between quasi-classical states, the conversion power of $\CTO$ is equal to the conversion power of $\GPO$~\cite{Janzing2000}. As in the previous proofs (see Eq.~\eqref{eq:relMajo} and the following discussion as well as Ref.~\cite{Renes2016}), this implies that
	\begin{align}
		&O_{\max}(\rho^R,\gamma^R,\gamma^A)= d\max\left\{t_1^A:(\vec{r}^R,\vec{g}^R)\succ (\tilde{\vec{t}}^A(t_1^A)^A,\vec{g}^A)\right\}.
	\end{align} 
	Due to the discussion around Eq.~\eqref{eq:conditionsTempChange}, we then have as a necessary condition (from $k=d$) that 
	\begin{align}
		O_{\max}(\rho^R,\gamma^R,\gamma^A)=d\tilde{t}_1^A\le\alpha_{dg_1^A}(\vec{r}^R,\vec{g}^R).
	\end{align} 
	Moreover, it holds that $O_{\max}(\rho^R,\gamma^R,\gamma^A)=\alpha_{g_1^A}(\vec{r}^R,\vec{g}^R)$, i.e., as large as possible according to this condition: If it holds, all conditions emerging from $k\ne d$ are satisfied too. This is apparent from geometrical considerations: the lower boundary of the testing region defined by Eq.~\eqref{eq:overlapMax} is fully specified by the three points $(0,0), (d\tilde{t}_1^A,dg_1^A),(1,1)$. Since the lower boundary of the testing region of $(\vec{r}^R,\vec{g}^R)$ is convex, having the "middle point" $(d\tilde{t}_1^A,dg_1^A)$ inside is enough to ensure that all conditions are satisfied (see Fig.~\ref{fig:altCooling}). 
	
	\begin{figure}[tb!]
		\centering
		\includegraphics[width=0.5\linewidth]{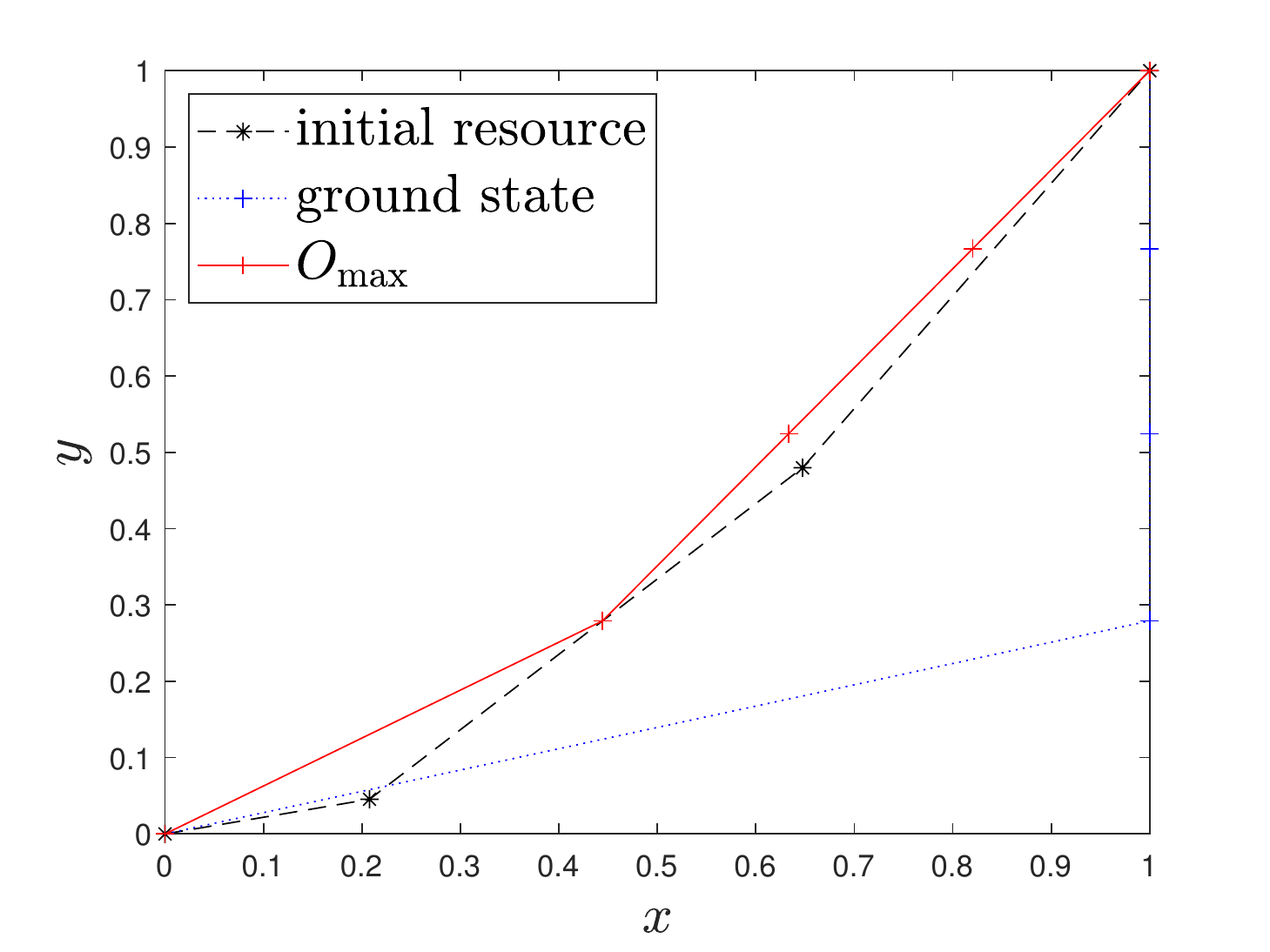}
		\caption{Lower boundaries of testing regions. The initial resource (black, dashed) of dimension three is used to cool a target system of dimension four (with a non-degenerate ground state). For the target system, the blue dotted line corresponds to the ground state, which would result in a perfect overlap. The ground state is not reachable, which can be seen by the fact that its lower boundary is below the lower boundary of the initial resource. The red solid line corresponds to a state with maximal achievable overlap, namely the one we used in the proof of Prop.~\ref{prop:altCooling}. }
		\label{fig:altCooling}
	\end{figure}

\end{proof}

We claimed in the main text that for qubits, the two interpretations of cooling coincide. As in the above proof, due to $\bra{1}\tau^A\ket{1}=\bra{1}\mathcal{P}_{\gamma^A}(\tau^A)\ket{1}$, where $\ket{1}$ denotes the in the qubit case non-degenerate ground state, we can restrict ourselves to quasi-classical target states when using the second interpretation. Since the ground state population and therefore the overlap with the ground state is monotonically increasing with decreasing temperature, we arrive at our claim: for qubits, maximizing the ground state overlap is equivalent to minimizing the temperature. The same argument shows that for qubits, also the two interpretations of heating coincide. For higher dimensional systems, they differ in general, as can be seen from a comparison of Thm.~\ref{thm:CoolingHeating} in the main text and Prop.~\ref{prop:altCooling}.

Moving on to the families of resource monotones defined in the main text, we discuss their properties, beginning with monotonicity  under $\CTO$.
If $(\rho^R,\gamma^R) \xrightarrow{\CTO} (\sigma^S,\gamma^S)$ we can use $(\rho^R,\gamma^R)$ to reach every $\tilde{\beta}$ that we can reach with $(\sigma^S,\gamma^S)$ by first converting $(\rho^R,\gamma^R)$ to $(\sigma^S,\gamma^S)$ and then using $(\sigma^S,\gamma^S)$ to reach $\tilde{\beta}$ (since $\CTO$ is closed under concatenation, see Lem.~\ref{lem:CTOclosedConc}). This implies monotonicity.
These monotones possess additional desirable properties that we now discuss at the example of $C_\beta^E(\rho,\gamma)$. Since $\tilde{\beta}_{\max}(\rho,\gamma;\beta,E)\ge\beta$ (we can always prepare a  Gibbs state corresponding to $\beta$ for free), $C_\beta^E(\rho,\gamma)\ge0$. Moreover, $C_\beta^E(\gamma,\gamma)=0$, since $\CTO$ is Gibbs-preserving. Due to Lem.~\ref{lem:targetQuasiClassical}, we note however that if $\mathcal{P}_{\gamma}(\rho)=\gamma$, then $C_\beta^E(\rho,\gamma)=0$, even if $\rho\ne\gamma$. On the contrary, if $\mathcal{P}_{\gamma}(\rho)\ne\gamma$, then $C_\beta^E(\rho,\gamma)>0$, because any non-free $(\vec{r},\vec{g})$ has a lower boundary of its testing region that is below the diagonal through the first quadrant, which, from a geometrical perspective, implies that we can use it to cool a qubit (see the proof of Thm.~\ref{thm:CoolingHeating} and Fig.~\ref{fig:Qubit}).
In summary, when restricted to quasi-classical states, $C_\beta^E(\rho,\gamma)$ is a faithful resource monotone, i.e., it is zero if and only if $\rho=\gamma$. For  $H_\beta^E(\rho,\gamma)$, it is straightforward to see that  analogous statements hold.

This provides us with the tools to prove Thm.~\ref{thm:main}.
\begin{thm}
	Let $(\sigma^S,\gamma^S)$ be quasi-classical. The following  statements are equivalent
	\begin{enumerate}
		\item $(\rho^R,\gamma^R) \xrightarrow{\CTO} (\sigma^S,\gamma^S)$.
		\item For any fixed $\beta>0$ and for all $E \in(0,\infty)$, it holds that
		\begin{align} 
			&C_\beta^E(\rho^R,\gamma^R)\ge C_\beta^E(\sigma^S,\gamma^S),
			\nonumber \\
			&H_\beta^E(\rho^R,\gamma^R)\ge H_\beta^E(\sigma^S,\gamma^S).
		\end{align}
	\end{enumerate}
\end{thm}
\begin{proof}
	First, we remember that according to Lem.~\ref{lem:targetQuasiClassical}, we can without loss of generality assume that $(\rho^R,\gamma^R)$ is quasi-classical too, since both $(\sigma^S,\gamma^S)$ and the Gibbs qubit states that we are considering are quasi-classical. As before, let thus $\vec{r}^R, \vec{g}^R, \vec{s}^S,\vec{g}^S$ denote the probability vectors corresponding to $\mathcal{P}_{\gamma^R}(\rho^R), \gamma^R,\sigma^S,\gamma^S$, respectively, and we can reduce our analysis to relative majorization.
	
	We first note that 2 follows from 1 due to the monotonicity of $C_\beta^E$ and $H_\beta^E$, see above.
	
	Let us next assume that 2 holds and note that according to the definitions in Eq.~\eqref{eq:Monotones} in the main text, this is equivalent to
	\begin{align}
		\tilde{\beta}_{\max}(\rho^R,\gamma^R;\beta,E)\ge& 	\tilde{\beta}_{\max}(\sigma^S,\gamma^S;\beta,E) \\
		\tilde{\beta}_{\min}(\rho^R,\gamma^R;\beta,E)\le& 	\tilde{\beta}_{\min}(\sigma^S,\gamma^S;\beta,E) 
	\end{align}	
	for the (by assumption) fixed $\beta$ and all $E \in(0,\infty)$.
	According to Eq.~\eqref{eq:elbowsCooling}, when cooling a qubit, the non-trivial elbow of the lower boundary of its testing region is given by 
	\begin{align}\label{eq:elbowQubitCooling}
		(\tilde{g}_1^A,g_1^A).
	\end{align}
	We first consider a fixed energy gap $E$, i.e., a fixed $g_1^A=g_1^A(\beta,E)$, and remember that how far we can cool the qubit using $(\rho^R,\gamma^R)$ is determined by the largest $\tilde{g}_1^A$ such that the elbow in Eq.~\eqref{eq:elbowQubitCooling} is still inside the testing region corresponding to $(\vec{r}^R,\vec{g}^R)$: 
	Using the argument that lead to Eq.~\eqref{eq:SMbetaMax}, we have that
	\begin{align}
		&\alpha_{ g_1^A(\beta,E)}(\vec{r}^R,\vec{g}^R)=\tilde{g}_1^A(\tilde{\beta}_{\max}(\rho^R,\gamma^R;\beta,E),E).
	\end{align}
	If for a given $E$, with $(\rho^R,\gamma^R)$ we can cool the qubit more than with $(\sigma^S,\gamma^S)$, we thus find
	\begin{align}\label{eq:RelCooling}
		&\alpha_{ g_1^A}(\vec{r}^R,\vec{g}^R)=\tilde{g}_1^A(\tilde{\beta}_{\max}(\rho^R,\gamma^R;\beta,E),E) 
		\ge\tilde{g}_1^A(\tilde{\beta}_{\max}(\sigma^S,\gamma^S;\beta,E),E)=\alpha_{ g_1^A}(\vec{s}^S,\vec{g}^S),
	\end{align}
	since $\tilde{g}_1$ increases if $\tilde{\beta}$ increases (we suppressed the dependence of $g_1^A$ on $\beta$ and $E$ for readability).
	We now note that by varying the energy gap $E$ between zero and infinity, we vary $g_1^A(\beta,E)$ between $1$ and $1/2$. 
	
	Next, we consider heating: According to Eq.~\eqref{eq:elbowsHeating}, the non-trivial elbow is given by 
	\begin{align}\label{eq:elbowQubitHeating}
		(\tilde{g}_2^A,g_2^A).
	\end{align}
	With fixed $E$, and analogous arguments, we then find that 
	\begin{align}\label{eq:RelHeating}
		&\alpha_{ g_2^A}(\vec{r}^R,\vec{g}^R)=\tilde{g}_2^A(\tilde{\beta}_{\min}(\rho^R,\gamma^R;\beta,E),E) 
		\ge\tilde{g}_2^A(\tilde{\beta}_{\min}(\sigma^S,\gamma^S;\beta,E),E)=\alpha_{ g_2^A}(\vec{s}^S,\vec{g}^S),
	\end{align}
	where we remember that we allowed for population inversions.
	This time, by varying $E$ between zero and infinity, we vary $g_2^A(\beta,E)$ between $0$ and $1/2$. 
	
	Combining the heating and cooling cases, we have shown that 
	\begin{align}
		&\alpha_{y}(\vec{r}^R,\vec{g}^R)\ge\alpha_{ y}(\vec{s}^S,\vec{g}^S)
	\end{align}
	for $y$ between $0$ and $1$. This implies relative majorization and therefore convertibility under $\CTO$, i.e., we have shown that 1 follows from 2.
	
\end{proof}
At this point, we want to remark that in Thm.~\ref{thm:main}, we could have added the third equivalent condition
\begin{enumerate}
	\setcounter{enumi}{2}
	\item \ Eqs.~\eqref{eq:thmConv1} hold for any fixed $E>0$ and for all $\beta \in(0,\infty)$.	
\end{enumerate} 
As in the previous proof, 3 follows from 1 due to the monotonicity of $C_\beta^E$ and $H_\beta^E$, and showing that 1 follows from 3 can be done completely analogously to the proof we used to show that 1 follows from 2: by keeping $E$ fixed and varying $\beta$ between zero and infinity, we vary $g_1^A(\beta,E)$ between 1 and 1/2 and $g_2^A(\beta,E)$ between $0$ and $1/2$ too. Whilst this statement is mathematically true, its physical interpretation is less clear: Changing the background temperature also affects the Gibbs states of any  systems $R$ and $S$ with fixed Hamiltonians.

We also note that for fixed $(\rho^R,\gamma^R), \ (\sigma^S,\gamma^S)$, in condition 2, it would be sufficient to demand that Eqs.~\eqref{eq:thmConv1} hold for a finite set of energy gaps $E$ if there exists no $k\in[|S|-1]$ such that $y_k^\star(\vec{s}^S,\vec{g}^S)=\frac{1}{2}$: Let $(\sigma^S,\gamma^S)$ be quasi-classical, $\beta>0$ fixed, 
\begin{align}
	\mathcal{S}_c:=&\left\{k\in[|S|-1]: y_k^\star(\vec{s}^S,\vec{g}^S)>\frac{1}{2}\right\} ,\nonumber \\
	\mathcal{S}_h:=&\left\{k\in[|S|-1]: y_k^\star(\vec{s}^S,\vec{g}^S)<\frac{1}{2}\right\}
\end{align}
and define the $|S|-1$ energies 
\begin{align}\label{eq:Ek}
	E_k=\begin{cases}
		&\frac{1}{\beta}\ln \frac{y_k^\star(\vec{s}^S,\vec{g}^S)}{1-y_k^\star(\vec{s}^S,\vec{g}^S)}  \text{ if } k\in \mathcal{S}_c,\\
		&\frac{1}{\beta} \ln \frac{1-y_k^\star(\vec{s}^S,\vec{g}^S)}{y_k^\star(\vec{s}^S,\vec{g}^S)}  \text{ if } k\in \mathcal{S}_h.
	\end{cases}
\end{align}	
By definition, $E_k>0$. We will now show that $(\rho^R,\gamma^R) \xrightarrow{\CTO} (\sigma^S,\gamma^S)$ iff
\begin{align} \label{eq:DiscMain}
	&C_\beta^{E_k}(\rho^R,\gamma^R)\ge C_\beta^{E_k}(\sigma^S,\gamma^S)\text{ for }k\in \mathcal{S}_c, \nonumber \\
	&H_\beta^{E_k}(\rho^R,\gamma^R)\ge H_\beta^{E_k}(\sigma^S,\gamma^S)\text{ for }k\in \mathcal{S}_h.
\end{align}
The forward direction follows directly from Thm.~\ref{thm:main}. For the reverse, assume that Eqs.~\eqref{eq:DiscMain} hold. According to Eq.~\eqref{eq:RelCooling}, this implies that for $k\in \mathcal{S}_c$,
\begin{align}
	&\alpha_{y_k^\star(\vec{s}^S,\vec{g}^S)}(\vec{r}^R,\vec{g}^R) =\alpha_{g_1^A(\beta,E_k)}(\vec{r}^R,\vec{g}^R)  
	\ge\alpha_{ g_1^A(\beta,E_k)}(\vec{s}^S,\vec{g}^S)=\alpha_{y_k^\star(\vec{s}^S,\vec{g}^S)}(\vec{s}^S,\vec{g}^S)=x_k^\star(\vec{s}^S,\vec{g}^S).
\end{align}
Moreover, according to Eq.~\eqref{eq:RelHeating}, it implies that 
\begin{align}
	\alpha_{y_k^\star(\vec{s}^S,\vec{g}^S)}(\vec{r}^R,\vec{g}^R) =\alpha_{g_2^A(\beta,E_k)}(\vec{r}^R,\vec{g}^R)  
	\ge&\alpha_{ g_2^A(\beta,E_k)}(\vec{s}^S,\vec{g}^S)=\alpha_{y_k^\star(\vec{s}^S,\vec{g}^S)}(\vec{s}^S,\vec{g}^S)=x_k^\star(\vec{s}^S,\vec{g}^S)
\end{align}
for $k\in \mathcal{S}_h$. In summary, we have thus shown that 
\begin{align}
	\alpha_{y_k^\star(\vec{s}^S,\vec{g}^S)}(\vec{r}^R,\vec{g}^R) 
	\ge x_k^\star(\vec{s}^S,\vec{g}^S)\ \forall k \in |S|-1.
\end{align}
According to Eq.~\eqref{eq:conditionsTempChange} (and noting that for $k=|S|$ the above condition is trivially satisfied) this implies relative majorization and therefore convertibility under $\CTO$, which finishes our proof. If there exists a $k\in[|S|-1]$ such that $y_k^\star(\vec{s}^S,\vec{g}^S)=\frac{1}{2}$, we can apply the above criteria to check convertibility to states arbitrarily close to the target state. Since CTO allows for an arbitrarily small error anyway, this is sufficient. 

To check convertibility to a single quasi-classical target state, using $|S|-1$ monotones is thus sufficient. However, which monotones we must choose depends on the specific target state (via the energies $E_k$ given in Eq.~\eqref{eq:Ek}). For different target states, we thus need different monotones and it is easy to see that a complete set of monotones that allows to check convertibility to an arbitrary quasi-classical target state necessarily includes all $E\in(0,\infty)$. This is not surprising, since it has been shown recently that at least in the limit of infinite temperature, there cannot exist a finite complete set of monotones~\cite{Datta2022}.

We now turn to the proof of Prop.~\ref{prop:gap} (in the main text), in parallel establishing some results that we also need for the proof of Thm.~\ref{thm:alternative}. Since we still investigate the heating and cooling of qubits, i.e., are concerned with quasi-classical target states, due to Lem.~\ref{lem:targetQuasiClassical}, we can again assume without loss of generality that $(\rho^R,\gamma^R)$, i.e., the initial resource, is quasi-classical too and reduce our analysis to relative majorization. To this end, let $\vec{r}^R,\vec{g}^R$ be the probability vectors corresponding to $\mathcal{P}(\rho^R),\gamma^R$ respectively.

Now let $A$ be a qubit with energy gap $E$. Its Gibbs state corresponding to the inverse background temperature $\beta$ is denoted by $\gamma^A$ and corresponds to the probability vector 
\begin{align}
	\vec{g}^A = 
	\begin{pmatrix}
		g^A_1\\1-g^A_1
	\end{pmatrix},\ g^A_1 = \frac{1}{1 + e^{-\beta E}}.
\end{align} 
Cooling or heating the qubit $A$ means that we want to create the athermality state $(\tilde{\gamma}^A,\gamma^A)$, where $\tilde{\gamma}$ is the Gibbs state corresponding to the target inverse temperature $\tilde{\beta}$, i.e.,
\begin{align}
	\tilde{\vec{g}}^A = 
	\begin{pmatrix}
		\tilde{g}_1^A\\1-\tilde{g}_1^A
	\end{pmatrix},\ \tilde{g}^A_1 = \frac{1}{1 + e^{-\tilde{\beta}E}}.
\end{align}

According to Eq.~\eqref{eq:elbowsCooling}, when cooling the qubit ($\tilde{\beta}>\beta$), the non-trivial elbow of the lower boundary of its testing region is given by 
\begin{align}
	(\tilde{g}_1^A,g_1^A)
\end{align}
and when heating (see Eq.~\eqref{eq:elbowsHeating}, $\tilde{\beta}<\beta$) by
\begin{align}
	(\tilde{g}_2^A,g_2^A).
\end{align}

Remembering that the background temperature is assumed to be finite and non-negative, i.e., $\beta>0$, which is physically well motivated, we introduce $w:=w(E):=e^{-\beta E}$, where $0<w\le1$, since without loss of generality, we only consider non-negative energy gaps. With $a=\tilde{\beta}/\beta$, this allows us to rewrite the elbows as
\begin{align}
	(\tilde{g}_1^A,g_1^A)=&\left(\frac{1}{1+e^{-\tilde{\beta}E\frac{\beta}{\beta}}},\frac{1}{1+e^{-\beta E}}\right) 
	=\left(\frac{1}{1+w^{a}},\frac{1}{1+w}\right), \nonumber \\
	(\tilde{g}_2^A,g_2^A)=&\left(\frac{w^a}{1+w^{a}},\frac{w}{1+w}\right).
\end{align}
In summary, the elbows of $(\tilde{\gamma}^A, \gamma^A)$ lie on the curve
\begin{align}
	F_a(w) = \begin{cases}
		\left(\frac{1}{1+w^a}, \frac{1}{1+w}\right) & \text{for }a > 1,\\
		\left(\frac{w^a}{1+w^a}, \frac{w}{1+w}\right) & \text{for }a < 1.\\
	\end{cases}
\end{align}
Note that the case $a=1$ is trivial, since it corresponds to leaving the qubit in its initial equilibrium state, which is of course always possible, independent of $E$.
\begin{figure}
	\centering
	\includegraphics[width=.5\linewidth]{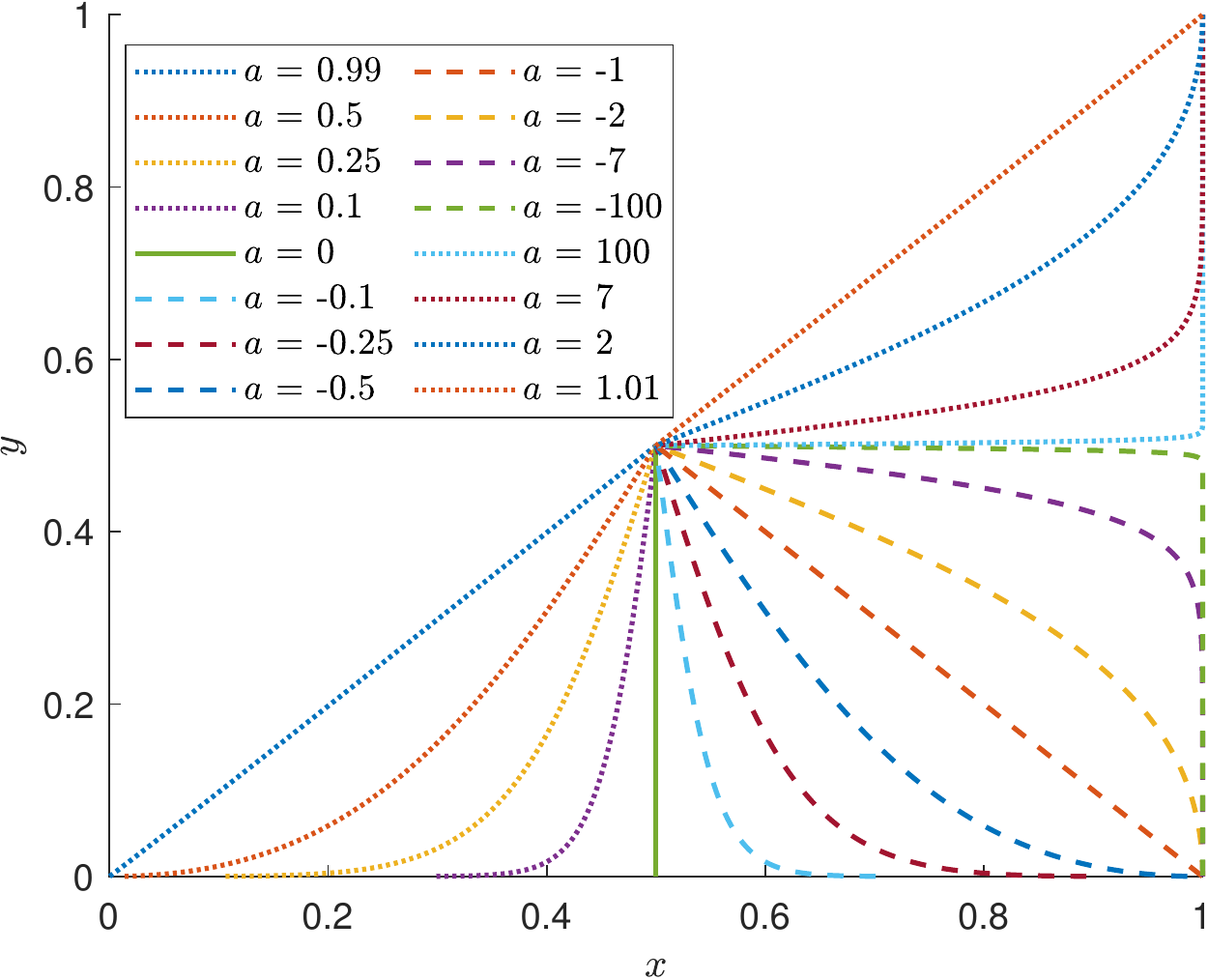}
	\caption{
		The curves $F_a(w)$ for various choices of $a\ne1$. Independent of $a$, $F_a(1)=(\frac{1}{2},\frac{1}{2})$. 
		Curves corresponding to $a>0$ and therefore finite positive target temperatures are dotted. The curves on the top right correspond to cooling and for $a \to 1^+$, (vanishing cooling) we approach a straight line from $(\frac{1}{2},\frac{1}{2})$ to $(0,1)$. The curves in the lower left correspond to heating, and for $a \to 1^-$ (vanishing heating), we approach a straight line from $(\frac{1}{2},\frac{1}{2})$ to $(0,0)$. The special case $a=0$ drawn as a solid straight line corresponds to heating to infinite temperature. Negative $a$ corresponding to population inversions in the target state are dashed. 
		For $a\to\pm\infty$, we approach vertical lines at $x=1$. 
	}\label{fig:Elbows}
\end{figure}

If $E=0$, then $w=1$, and $F_a(1) = \left(\frac{1}{2},\frac{1}{2}\right)$. If instead $E\to\infty$, then $w\to0^+$, and
\begin{align}
	\lim_{w\to0^+}F_a(w) = \begin{cases}
		(1,1) & \text{for }a >1,\\
		(0,0) & \text{for }0<a <1,\\
		\left(\frac{1}{2},0\right) & \text{for }a=0,\\
		(1,0) & \text{for }a <0.
	\end{cases}
\end{align} 
For all fixed $a\ne1$, the continuous curves describing the position of the elbows in terms of $E$ thus start at $\left(\frac{1}{2},\frac{1}{2}\right)$ (corresponding to $E =0$) and go towards the boundary of the square with corners $(0,0),(0,1),(1,1),(1,0)$ for $E \to \infty$ (see Fig.~\ref{fig:Elbows}).

From Fig.~\ref{fig:Elbows}, it is plausible that for $a > 0$, $a \neq 1$ there exist (quasi-classical) athermality states $(\rho^R, \gamma^R)$ with lower boundaries of their associated testing regions that cross the curve $F_a(w)$ in multiple points. 
In the following, we provide an explicit construction of such states. Since $F_a(w)$ determines the position of the qubit's elbow in terms of $E$ and $\tilde{\beta}$, this will imply Prop.~\ref{prop:gap} (in the main text) due to the connection of (quasi-classical) state transformations and relative majorization.

First, we consider the case $0<a<1$. Let $f_a(x)$ be the affine function that passes through  $(1,1)$ and is tangent to $F_{a}(w)$. As seen in Fig.~\ref{fig:Elbows}, this function always exists.
Denoting by $(x_0, f_a(x_0))$ the tangent point, clearly $0 < x_0, f_a(x_0) < \frac{1}{2}$ and $f_a(0) < 0$.
Let now $\tilde{f}_a(x)$ be the affine function passing through $(1,1)$ and $(0, f_a(0)/2)$. By construction, this function crosses $F_a(w)$ in two points, and we denote the corresponding $x$-values by $x_2<x_3$. Furthermore, $\tilde{f}_a(x)$ crosses the $x$-axis at $x_4<x_2$. With $x_1=(x_2-x_4)/2$, the qubit athermality state with corresponding elbow $(x_1, \tilde{f}_a(x_1))$ satisfies our claims. Moreover, it can be shown that for this elbow, there is a third crossing for an $x<x_1$, see Fig.~\ref{fig:ExampleGap} for the case $a=1/2$. As one might see again from Fig.~\ref{fig:Elbows}, for $a\to1^-$, this is not apparent. Therefore, more formally, by setting $x=\frac{w^a}{1+w^a}$, we can give the $y$ values of the curve $F_a(w)$ as a function of the $x$ values,
\begin{align}
	\tilde{F}_a(x)=\frac{x^{1/a}}{(1-x)^{1/a}+x^{1/a}}.
\end{align}
Then 
\begin{align}
	\tilde{F}'_a(x)=&\frac{1}{a}\frac{1}{(1-x)^{1/a}+x^{1/a}} \left(x^{1/a-1}-x^{1/a}\frac{x^{1/a-1}-(1-x)^{1/a-1}}{(1-x)^{1/a}+x^{1/a}}\right)
\end{align}
and (remember $0<a<1$)
\begin{align}
	\lim_{x\to0^+}\tilde{F}'_a(x)=0,
\end{align}
from which follows the claim. Notice that this also implies that for $E\to\infty$, heating to any non-trivial target temperature becomes impossible.

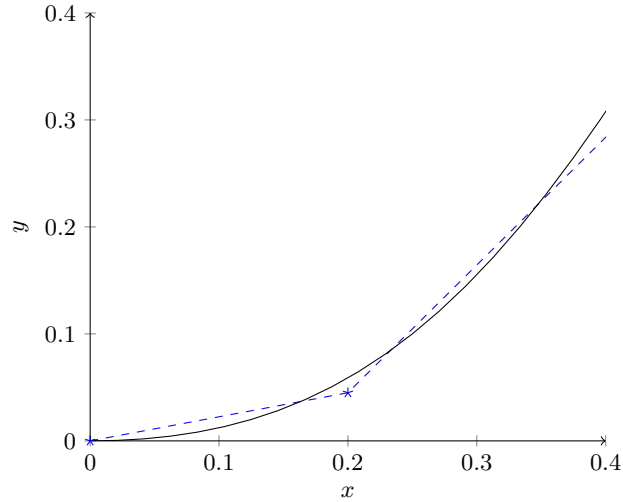
\begin{figure}
	\centering
	\begin{tikzpicture}
		\begin{axis}[xmin=0,xmax=0.4,ymin=0,ymax=0.4, axis x 		line=bottom,axis y line=left,
			axis line style={->},
			ylabel near ticks,
			xlabel near ticks,
			xlabel={$x$ },
			ylabel={$y$ }]
			\addplot+[sharp plot, blue, dashed, mark = star] coordinates {(0,0)(0.2,0.045) (1,1)};
			\addplot[domain=0:1/2] {x^2/((1-x)^2+x^2)};
		\end{axis}
	\end{tikzpicture}
	\caption{The lower boundary of the curve associated with a qubit athermality state $(\rho^R, \gamma^R)$ (blue, dashed) which intersects $F_{1/2}(w)$ (black, solid) three times.  }\label{fig:ExampleGap}
\end{figure}

For higher dimensional systems $R$, the same construction is possible (e.g., by an effective reduction to one elbow such as shown in Fig.~\ref{fig:altCooling}). It is however also possible to construct additional intersections: Starting as in the qubit case, one determines $(x_1, \tilde{f}_a(x_1))$ and then considers the affine function that goes through that point and is tangent to $F_a(w)$ at a point $x<x_1$. Analogously, one can then construct another elbow leading to additional intersections. From the maximal number of elbows being $|R|-1$ follows the maximal possible number of intersections. 

We now turn to the case of $a>1$. Apparently, an analogous construction is possible, now starting with the affine function that passes through $(0,0)$ and is tangent to $F_a(w)$. In summary, this proves Prop.~\eqref{prop:gap} in the main text.

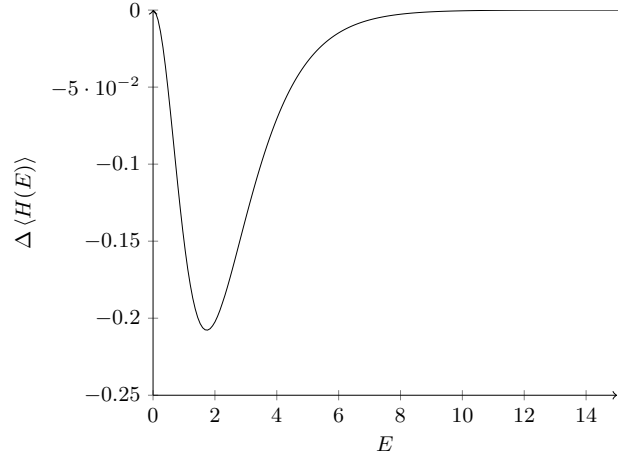
\begin{figure}
	\centering
	\begin{tikzpicture}[scale=0.9]
		\begin{axis}[xmin=0,xmax=15
			,ymin=-0.25,ymax=0.0, axis x line=bottom,axis y line=left,
			axis line style={->},
			ylabel near ticks,
			xlabel near ticks,
			xlabel={$E$ },
			ylabel={$\Delta\braket{H(E)}$ }]
			\addplot[domain=0:40, samples=500] {(e^(-2*x)/(1+e^(-2*x))-e^(-x)/(1+e^(-x)))*x};
		\end{axis}
	\end{tikzpicture}
	\caption{The change of the expectation value of the energy $\Delta\braket{H(E)}$ of a a qubit when cooling it from $\beta=1$ to $\tilde{\beta}=2$ against its energy gap $E$.}\label{fig:EnergyChange}
\end{figure}
As mentioned in the main text, Prop.~\ref{prop:gap} seems to be surprising at first.  However, the change of the expectation value of the energy of a qubit with energy gap $E$ when cooling/heating it from $\beta$ to $\tilde{\beta}$ is given by 
\begin{align}
	\Delta\braket{H(E)}=\left(\frac{e^{-\tilde{\beta}E}}{1+e^{-\tilde{\beta}E}}-\frac{e^{-\beta E}}{1+e^{-\beta E}}\right)E,
\end{align}
which, as shown in Fig.~\ref{fig:EnergyChange}, is not a monotonic function of $E$. This provides an intuition for Prop.~\eqref{prop:gap} (in the main text): If the initial resource cannot provide/absorb enough energy, heating/cooling to a given $\tilde{\beta}$ is impossible. However, it is well known that in the quantum regime, many second laws exist~\cite{Brandao2015}, i.e., (free) energy considerations alone are not sufficient to determine if a transformation is possible or not. This explains why according to Fig.~\ref{fig:EnergyChange}, there are no energy restrictions if $E\to\infty$, but as discussed earlier, for $E\to\infty$, heating becomes impossible. 

We conclude by proving Thm.~\ref{thm:alternative}.
\setcounter{theorems}{5}
\begin{thm}
	Let $(\sigma^S, \gamma^S)$ be quasi-classical. 
	Then $(\rho^R, \gamma^R) \xrightarrow{\CTO} (\sigma^S, \gamma^S)$ if and only if 
	\begin{align}
		\mathfrak{E}_{\beta}{(\rho^R, \gamma^R;\tilde{\beta})} \supseteq \mathfrak{E}_{\beta}{(\sigma^S, \gamma^S;\tilde{\beta})}
	\end{align}
	for any fixed $\beta>0$ and all $\tilde{\beta}\in(-\infty,\infty)$. 
\end{thm}

\begin{proof}
	One direction is again trivial. If 
	\begin{align}
		(\rho^R, \gamma^R) \xrightarrow{\CTO} (\sigma^S, \gamma^S)
	\end{align}
	and 
	\begin{align}
		(\sigma^S, \gamma^S) \xrightarrow{\CTO} (\tilde{\gamma}^A, \gamma^A)
	\end{align}
	then also 
	\begin{align}
		(\rho^R, \gamma^R) \xrightarrow{\CTO} (\tilde{\gamma}^A, \gamma^A),
	\end{align}
	since $\CTO$ is closed under concatenation.
	Thus $\mathfrak{E}_{\beta}{(\rho^R, \gamma^R;\tilde{\beta})} \supseteq \mathfrak{E}_{\beta}{(\sigma^S, \gamma^S;\tilde{\beta})}$. Moreover, due to Lem.~\ref{lem:targetQuasiClassical}, we can again assume without loss of generality that $(\rho^R,\gamma^R)$ is quasi-classical too, and identify $\vec{r}^R, \vec{g}^R, \vec{s}^S,\vec{g}^S$ with $\mathcal{P}_{\gamma^R}(\rho^R), \gamma^R,\sigma^S,\gamma^S$, respectively. 
	
	Assume now that $\mathfrak{E}_{\beta}{(\rho^R, \gamma^R;\tilde{\beta})} \supseteq \mathfrak{E}_{\beta}{(\sigma^S, \gamma^S;\tilde{\beta})}$ for a fixed $\beta>0$ and all $\tilde{\beta}\in(-\infty,\infty)$ and remember that 
	\begin{align}
		(\rho^R, \gamma^R) \xrightarrow{\CTO} (\sigma^S, \gamma^S)
	\end{align}
	if and only if $(\vec{r}^R,\vec{g}^R)\succ (\vec{s}^S,\vec{g}^S)$, i.e., if the testing region associated with $(\vec{r}^R,\vec{g}^R)$ contains the testing region associated with $(\vec{s}^S,\vec{g}^S)$. Now suppose that the testing region associated with $(\vec{s}^S,\vec{g}^S)$ is not contained in the testing region associated with $(\vec{r}^R,\vec{g}^R)$, i.e., that there exists at least one point $(x_0,y_0)$ outside the testing region corresponding to $(\vec{r}^R,\vec{g}^R)$ and inside the testing region corresponding to $(\vec{s}^S,\vec{g}^S)$. 
	
	Using the notation introduced previously, we then choose $a_0,w_0$ such that $(x_0,y_0)=F_{a_0}(w_0)$. For $y_0\notin\{0,1/2\}$ and $x_0\ne 1 $, such a pair always exists, see Fig.~\ref{fig:Elbows}. If $y_0\in\{0,1/2\}$ or $x_0= 1 $, by continuity of the boundaries of the testing regions, there always exists a point $(x_0',y_0')$ close to $(x_0,y_0)$ that also satisfies that it is outside of the testing region corresponding to $(\vec{r}^R,\vec{g}^R)$ and inside the testing region corresponding to $(\vec{s}^S,\vec{g}^S)$ such that $y_0'\notin\{0,1/2\}$ and $x_0'\ne 1 $ and we use that point for our argument instead.
	
	With $E_0$ such that $w_0= e^{-\beta E_0}$ and $\tilde{\beta}_0$ such that $a_0=\tilde{\beta_0}/\beta$, we then find that $E_0\notin\mathfrak{E}_{\beta}{(\rho^R, \gamma^R;\tilde{\beta}_0)}$ and $ E_0\in \mathfrak{E}_{\beta}{(\sigma^S, \gamma^S;\tilde{\beta}_0)}$. This is clearly a contradiction to the assumption  that $\mathfrak{E}_{\beta}{(\rho^R, \gamma^R;\tilde{\beta})} \supseteq \mathfrak{E}_{\beta}{(\sigma^S, \gamma^S;\tilde{\beta})}$ for all $\tilde{\beta}\in(-\infty,\infty)$.	
\end{proof}

\end{document}